\documentclass[sort&compress]{elsarticle}

\usepackage{latexsym}
\usepackage{epsfig}
\usepackage{amsmath}
\usepackage{amsfonts}
\usepackage{amssymb}
\usepackage{graphicx}
\usepackage{pgf}
\usepackage{enumerate}
\usepackage{amssymb}
\usepackage{caption}
\usepackage{subcaption}
\usepackage{lineno}
\usepackage{multicol}
\usepackage{mathtools}
\mathtoolsset{showonlyrefs}

\usepackage[utf8]{inputenc}
\usepackage{amsthm}
\usepackage[english]{babel}
\usepackage{esint}
\usepackage{fontenc}
\usepackage{fullpage}
\usepackage{color}

\newcommand{\ls}{\leqslant}
\newcommand{\gs}{\geqslant}

\newcommand{\R}{\mathbb R}
\DeclareRobustCommand{\rchi}{{\mathpalette\irchi\relax}}
\newcommand{\irchi}[2]{\raisebox{\depth}{$#1\chi$}} 
\newcommand{\dd}{\, \mathrm{d}}
\renewcommand{\d}{\mathrm{d}}
\newcommand{\eps}{\varepsilon}
\newcommand{\tr}{\operatorname{Tr}}
\newcommand{\BD}{\operatorname{BD}}

\newcommand{\TD}{\operatorname{TD}}
\newcommand{\Id}{\operatorname{I}}
\renewcommand{\div}{\operatorname{div}}
\newcommand{\curl}{\operatorname{curl}}

\newcommand{\scal}[2]{\left \langle #1,#2 \right \rangle}

\newcommand{\gsym}{\dot{\bs \gamma}}
\newtheorem{thm}{Theorem}
\newtheorem{lem}{Lemma}

\newtheorem{cor}{Corollary}

\newcommand{\bs}{\boldsymbol}

\makeatletter
\def\ps@pprintTitle{%
 \let\@oddhead\@empty
 \let\@evenhead\@empty
 \let\@oddfoot\@empty
 \let\@evenfoot\@oddfoot
}
\makeatother

\begin{document}

\begin{frontmatter}

\title{Computing the Yield Limit in Three-dimensional Flows of a Yield Stress Fluid About a Settling Particle}

\author[Linz]{Jos\'{e} A. Iglesias\corref{correspondingauthor}}
\author[Vienna]{Gwenael Mercier}
\author[mainaddress]{Emad Chaparian}
\author[mainaddress]{Ian A. Frigaard}
\address[Linz]{Johann Radon Institute for Computational and Applied Mathematics (RICAM), Austrian Academy of Sciences, Linz, Austria}
\address[Vienna]{Faculty of Mathematics, University of Vienna, Vienna, Austria}
\address[mainaddress]{Department of Mechanical Engineering \& Department of Mathematics, University of British Columbia, Vancouver, BC, V6T 1Z4, Canada}

\cortext[correspondingauthor]{Corresponding author: jose.iglesias@ricam.oeaw.ac.at}

\begin{abstract}
Calculating the yield limit {$Y_c$ (the critical ratio of the yield stress to the {\it driving} stress)}, of a viscoplastic fluid flow is a challenging problem, often needing iteration in the rheological parameters to approach this limit, as well as accurate computations that account properly for the yield stress and potentially adaptive meshing. For particle settling flows, in recent years calculating $Y_c$  has been accomplished analytically for many antiplane shear flow configurations and also computationally for many geometries, under either two dimensional (2D) or axisymmetric flow restrictions. Here we approach the problem of 3D particle settling and how to compute the yield limit \textbf{directly}, i.e.~without iteratively changing the rheology to approach the yield limit. The presented approach develops tools from optimization theory, taking advantage of the fact that $Y_c$ is defined via a minimization problem. We recast this minimization in terms of primal and dual variational problems, develop the necessary theory and finally implement a basic but workable algorithm. We benchmark results against accurate axisymmetric flow computations for cylinders and ellipsoids{, computed using adaptive meshing. We also make comparisons of accuracy in calculating $Y_c$ on comparable fixed meshes.} This demonstrates the feasibility {and benefits} of directly computing $Y_c$ in multiple dimensions. {Lastly, we present some sample computations for complex 3D particle shapes.}
\end{abstract}


\end{frontmatter}


\section{Introduction}
\label{sec:intro}

The ability to resist a shear stress at rest is the key qualitative feature of a yield stress fluid. This feature arises in the earliest classical flows studied with simple yield stress fluid models: for sufficiently large yield stresses: dense particles are statically suspended in fluid, a layer of  yield stress fluid on a vertical surface does not flow, a stress-controlled vane viscometer cannot turn, etc. This competition is captured physically as the dimensionless balance $Y$, between the yield stress and whatever is the driving force (stress) of the flow, e.g.~pressure drop in flow along a pipe, buoyancy in settling of a particle or in bubble rise. Many 1D flows allow analysis and explicit determination of limiting values of the yield number $Y$, such that there is no flow for $Y \ge Y_c$. General determination of $Y_c$ remains elusive although $Y_c$ has been given a  formal mathematical definition for various classes of flow e.g.~\cite{putz2010,karimfazli2016}, the roots of which go back to the 1965 analysis of \cite{MosMia1965} for anti-plane shear flows.

In this paper our concern is computation of $Y_c$ for 3D settling particles and more specifically \textbf{direct} computation. To explain direct, we first discuss more general particle computations for yield stress fluids.
The best known work here by Beris \emph{et al.}~\cite{beris1985} considered flow around a sphere and was the first to convincingly frame the theoretical question of a limiting $Y$ and to calculate it. There have since been extensive calculations of flows around isolated axisymmetric and 2D particles e.g.~\cite{beaulne1997,zisis2002,degloDeBesses2003,degloDeBesses2004,mitsoulis2004a,tokpavi2008,jossic2009,putz2010,chaparian2017yield,chaparian2017cloaking}. For Stokes flows we must distinguish 2 different formulations for these problems: a resistance problem and a mobility problem. Simply put, a resistance problem imposes motion on a particle and computes the force \& moment; a mobility problem imposes the forces/moments and computes velocity. Either formulation can be used to estimate $Y_c$. For the mobility formulation, for fixed imposed forces/moments one increases the yield stress incrementally until a zero velocity is computed. For the resistance formulation an asymptotic approach is needed, via rescaling of variables, again with successive computations; see e.g.~\cite{chaparian2017yield}.

The need to compute a flow at successively large values of a dimensionless parameter makes these methods iterative or indirect, in terms of finding $Y_c$. The limiting flows are generally characterized by diminishing relevance of the viscous dissipation with respect to the plastic dissipation. This leads naturally to the question of whether we can simply neglect the viscous aspect of the flow \emph{a priori}, whether such a computation is viable and whether the solution allows us to compute the limit $Y_c$ \textbf{directly}. The motivation for this is partly intellectual and partly pragmatic. Calculation of $Y_c$ for 2D and axisymmetric flows using indirect methods has been most convincing when an augmented Lagrangian method has been used and when some form of mesh adaptivity is implemented, so that the mesh can deform to relevant geometric features of the limiting flow. This approach was originally developed by Saramito and co-workers \cite{roquet2003} and has been used extensively over the past decade for particle settling, e.g.~\cite{putz2010,chaparian2017yield}. Although highly effective, the mesh refinement requires successive computation and each flow computation using augmented Lagrangian approaches is itself slow to converge iteratively. While viable in 2D/axisymmetric situations, moving to 3D has not been explored. Indeed there are relatively few 3D particle computations for viscoplastic fluids at all and none that we know with adaptive meshing and augmented Lagrangian methods. Coupling this to an iterative parametric increase to attain $Y_c$ is likely to be prohibitive computationally. { Secondly, numerical implementation of the yielding problem via Mobility ([M]) formulations to calculate $Y_c$ is not always trivial which leads to the alternative usage of Resistance ([R]) formulations to study this limit (the definitions of [R] and [M] formulations are in section \ref{sec:problem}). Investigating the yield limit using the [R] formulations has its own drawbacks as well. For instance, calculating the yield limit needs some computations at moderate/finite yield stresses and then extrapolating to an infinitely large yield stress which is: (i) time-consuming; although some new accelerating augmented Lagrangian approaches have been proposed recently such as FISTA \cite{TreRouFriWac18} and PAL method \cite{dimakopoulos2018pal} and (ii) sensitive to extrapolation procedure as well (for instance, see Fig.~23 of \cite{tokpavi2008})}. Thus, direct methods become attractive.

The idea of a direct calculation is not new. For 1D shear flows, we often first compute the shear stress and then the velocity. If the shear stress does not exceed the yield stress there is no flow and no need to compute the velocity: a direct calculation of $Y_c$, e.g. for flow in a channel/pipe we require the wall shear stress to exceed the yield stress. The analysis of Mosolov \& Miasnikov \cite{MosMia1965} turned the limiting variational problem (velocity minimization) into a geometric problem, solved directly for $Y_c$: see the almost algorithmic application of this in \cite{huilgol2006}. In \cite{frigaard2017} we have extended the scope of \cite{MosMia1965} to anti-plane particle settling flows, and used a set theoretic approach which again directly calculates $Y_c$ without recourse to the viscous solutions. More general anti-plane flows, including non-convex and multiple particle arrays were considered in \cite{IglMerSch18b}, again for anti-plane shear flows.

In moving to 2D flows there are physically based direct approaches such as the theory of perfect plasticity, which involves constructing a network of sliplines around a particle. As explored in \cite{chaparian2017yield} the slipline method does not always produce velocity and stress fields that are representative on the limiting flow, although sometimes the estimates of $Y_c$ are exact. The latter often happens when the flow has a specific geometric structure that allows one to estimate terms in the optimization problem directly. Thus, for example in \cite{karimfazli2015} flow onset in a convectively driven flow was observed as yielding along the boundary of a cylindrical domain; an observation that allowed direct calculation of $Y_c$, albeit heuristically.

In the broader context, there are other reasons for calculating $Y_c$. First, as well as being a limiting value for the static steady problem, for many flows $Y \ge Y_c$ also means that the associated unsteady flows converge to the static steady state, i.e.~stability; see \cite{karimfazli2016,Wachs2016}. This type of result, originally established by \cite{bristeau1975} for 1D channel flow, has much wider application, as discussed in \cite{frigaard2019}.
A long term goal of work in this area is to be able to characterize statically stable limits for viscoplastic suspensions i.e.~particles in a simple yield stress fluid. Industrially used suspensions are often buoyant and it is of some value to understand idealized behaviours. There are a number of computational studies that deal with multiple particles, e.g.~\cite{koblitz2018a,liu2003,tokpavi2009,prashant2011,chaparian2018inline,fahs2018} in both steady and transient situations. Ostensibly these studies are focused at the micro-scale hydrodynamics, and they do uncover details of yielding, bridging between particles and contact. However, it is not clear what the next step is for such studies: unlike Newtonian fluids, the viscoplastic Stokes equations are not linear in the velocity, so superposition principles fail and methods such as Stokesian dynamics are not valid. A pure computational approach has been followed in \cite{yu2007,Wachs2016} in which a fictitious domain approach is used. This approach models the full fluid-solid domain and is implemented within the augmented lagrangian framework. Although suited to suspensions the computational algorithms are relatively slow to converge at each timestep and transient calculations have thus been limited to small numbers of particles in two dimensions. Recent work \cite{koblitz2018b} has looked at steady Stokes flow of 2D particles arranged in random configurations, to represent a given solids fraction. These calculations are interesting in that, by repeating with successively large yield stress, we begin to estimate the critical yield limit $Y_c$ of a suspension, i.e.~at which the entire suspension is static. Also the relative economy of the Stokes flow calculations means that a good level of mesh refinement is feasible. We are beginning to uncover some of the true complexity of the stress field in suspensions.

An outline of the paper is as follows. Below (\S \ref{sec:problem}) we introduce the physical problem and give the formal definition of $Y_c$. We show that the primal problem has a solution in a subspace of the {space of vector fields of bounded deformation BD} (Theorems \ref{thm:compact} \& \ref{thm:relaxation}) and then explore two dual problem formulations. Section \ref{sec:computing} explains the computational algorithms for boththe direct method (\S \ref{sec:minimization}) and for axisymmetric particle computations carried out using a more conventional method (\S \ref{sec:axisymmetric_calcs}). In \S \ref{sec:results} we present a range of results for flows around 3D particles (axisymmetric or not, convex or not). The paper finishes with a discussion and conclusions.

\section{Problem statement}
\label{sec:problem}

In this paper we focus on the motion of an isolated particle in a \emph{large} bath of yield-stress fluid. We are specifically interested in computing the static stability of particles or the yield limit, i.e.~when the force on the particle is just enough to move it. Hence we only consider inertia-less flows. The particle is denoted by $X$, $\partial X$ is the boundary of the particle, $\Omega$ represents the entire domain (fluid+particle) and $\partial \Omega$ is its outer boundary.

For any fixed finite yield stress and body force on the particle (e.g.~buoyancy), the deviatoric stress is expected to decay at large distances from the particle, eventually dropping below the yield stress. The fluid is observed to become unyielded and static at a sufficient distance from a particle. Thus, fixing $\Omega \subset \mathbb{R}^3$ to be any set sufficiently \emph{large} for the flow to be static at $\partial \Omega$ produces an equivalent velocity field.

The flow problem may be formulated in two ways, as follows:
\begin{enumerate}[(i)]
\item Mobility problem [M]: in which the particle is driven by a body force, e.g.~sedimentation under gravitational force when the particle is denser than the fluid ($\hat{\rho}_p > \hat{\rho}_f$), or rising when buoyant. Here gravitational acceleration is aligned with the negative $z$-direction ($\hat{\boldsymbol{g}} = - \hat{g} ~\boldsymbol{e}_z$) and we assume that the {particle $X$ has enough symmetries} (for example a symmetry about the $z$-axis, but this applies to more situations described below in \S \ref{sec:minimization}), so that there is no rotational torque exerted on the particle; see \cite{putz2010}. The traction on the particle surface satisfies:
\begin{equation}
\int_{\partial X} \boldsymbol{\sigma} \cdot \boldsymbol{n}_p~ ds = -\int_{\partial X} \boldsymbol{\sigma} \cdot \boldsymbol{n}~ ds = -\frac{V_p}{1-\rho} \boldsymbol{e}_z,
\label{eq:bc2}
\end{equation}
where $\boldsymbol{n}_p (= -\boldsymbol{n})$ is the outward normal to the particle, $V_p$ is the dimensionless volume of the particle and $\rho = \left( \hat{\rho}_f/\hat{\rho}_p \right) < 1$ is the density ratio. Dimensional quantities are distinguished with a $\hat{\cdot}$ accent, i.e.~$\hat{\rho}_p$ and $\hat{\rho}_f$ are the densities of particle and fluid, respectively. The Cauchy stress tensor, $\boldsymbol{\sigma} = -p \boldsymbol{\delta} + \boldsymbol{\tau}$, satisfies the dimensionless Stokes  equations:
\begin{equation}\label{eq:M_problem}
\boldsymbol{\nabla} \boldsymbol{\cdot} \boldsymbol{\sigma} + \frac{\rho}{1-\rho} \boldsymbol{e}_z = 0,~~~~\boldsymbol{\nabla} \boldsymbol{\cdot} \boldsymbol{u} = 0 \quad \mbox{in}\quad \Omega \setminus \bar{X} ,
\end{equation}
where $\boldsymbol{u}$ is the fluid velocity.
The deviatoric stress is defined by the constitutive equations for a Bingham fluid:
\begin{equation}\label{eq:Cons_M_Bingham}
\left\{
\begin{array}{ll}
\boldsymbol{\tau} = \left( 1 + \displaystyle{\frac{Y}{\vert\dot{\boldsymbol{\gamma}}\vert_F}} \right) \dot{\boldsymbol{\gamma}} & \text{if}~~ \vert\boldsymbol{\tau}\vert_F > Y, \\[2pt]
\dot{\boldsymbol{\gamma}} = 0 & \text{if}~~ \vert\boldsymbol{\tau}\vert_F \leqslant Y,
\end{array} \right.
\end{equation}
where $Y=\hat{\tau}_Y/(\hat{\rho}_p - \hat{\rho}_f) \hat{g} \hat{L}$ is the yield number and $\hat{\tau}_Y$ is the yield stress of the fluid. We define the length-scale $\hat{L}$ later.
The velocity $\boldsymbol{u}$ vanishes in the far-field, is continuous at the particle surface (no-slip) and the stress at the particle surface satisfies \eqref{eq:bc2}.
The strain rate tensor
\[\dot{\boldsymbol{\gamma}}(\bs u) = \partial_i u^j  + \partial_j u^i\]
and the tensor norm $\vert \cdot \vert_F$ is the Frobenius norm associated with the inner product:
\[ \boldsymbol{a}:\boldsymbol{b} = \frac{1}{2}\sum_{i,j=1}^3 a_{ij}b_{ij},~~~~~~
  \vert\boldsymbol{a}\vert_F = (\boldsymbol{a}:\boldsymbol{a})^{1/2} . \]

\item Resistance problem [R]: in which the problem is defined based on an imposed particle motion, e.g.~$\hat{\boldsymbol{U}}^* = - \hat{U}^* ~\boldsymbol{e}_z$, giving a Dirichlet boundary condition on the particle surface. The imposed velocity is used to define a stress scale for the flow and the dimensionless field equations are:
\begin{equation}\label{eq:R_problem}
\boldsymbol{\nabla} \boldsymbol{\cdot} \boldsymbol{\sigma}^* = 0,~~~~\boldsymbol{\nabla} \boldsymbol{\cdot} \boldsymbol{u}^* = 0 \quad \mbox{in}\quad \Omega \setminus \bar{X},
\end{equation}
where $\boldsymbol{u}^*$ is the velocity and $\boldsymbol{\sigma}^* \left(= -p^* \boldsymbol{\delta} + \boldsymbol{\tau}^* \right)$ is the Cauchy stress tensor. The constitutive law is:
\begin{equation}\label{eq:Cons_R_Bingham}
\left\{
\begin{array}{ll}
\boldsymbol{\tau}^* = \left( 1 + \displaystyle{\frac{B}{\vert\dot{\boldsymbol{\gamma}}^*\vert_F}} \right) \dot{\boldsymbol{\gamma}}^* & \text{iff}~~ \vert\boldsymbol{\tau}^*\vert_F > B, \\[2pt]
\dot{\boldsymbol{\gamma}}^* = 0 & \text{iff}~~ \vert\boldsymbol{\tau}^*\vert_F \leqslant B,
\end{array} \right.
\end{equation}
where $B=\hat{\tau}_Y \hat{L}/\hat{\mu} \hat{U}^*$ is the Bingham number and $\hat{\mu}$ is the plastic viscosity of the fluid. The velocity vanishes in the far-field and $\boldsymbol{u}^*$ is specified at the particle surface.
\end{enumerate}
Derivation of [M] and [R] is described in more detail in \cite{putz2010} as is the 1-to-1 relationship between the two problems (provided that the flow is non-zero). In \cite{chaparian2017yield} we have addressed the yield limit problem for the case of 2D planar flows around symmetric particles. Either [M] or [R] problems can be used to study the static stability of the particle (or yield limit). Considering [M], intuitively we expect that for sufficiently large yield stress, acting over a characteristic yield surface, the buoyancy force will be balanced and the motion is arrested. Since this ratio is captured by $Y$, for the mobility problem the static stability limit is simply $Y \to Y_c^-$.

The correspondence between [M] or [R] problems comes through the plastic drag coefficient $C_d^P$,  which is defined for both problems [M] and [R] as follows:
\begin{equation}
C_d^P = \left\{
\begin{array}{ll}
\displaystyle{\left[ \frac{\hat{F}^*}{\hat{A}_\bot \hat{\tau}_Y} \right]_{[R]} = \left[ \frac{F^*}{A_\bot B} \right]_{[R]}} & \mbox{for problem [R]}, \\[2pt]
\displaystyle{\left[ \frac{ [\hat{\rho}_p- \hat{\rho}_f] \hat{g} \hat{V}_p}{\hat{A}_\bot \hat{\tau}_Y} \right]_{[M]} = \left[ \frac{V_p}{A_\bot Y} \right]_{[M]}} & \mbox{for problem [M]} ,
\end{array} \right.
\end{equation}
see e.g.~\cite{chaparian2017yield,chaparian2017cloaking}. Here $\hat{F}^*$ is the force on the particle and $\hat{A}_\bot$ represents the frontal area of the particle, perpendicular to the direction of motion.

For static stability, in an [R] problem, the particle never stops, because we always impose a dimensionless velocity of unit magnitude. However, in the limit of $B \to \infty$, motion becomes increasingly difficult. In this limit, the particle force $F^* \to \infty$ and the motion is  localised adjacent to the particle surface. The ratio $F^* / B$ asymptotes to a constant value, which is the critical plastic drag coefficient times the frontal area of the particle. The critical plastic drag coefficient therefore represents the equivalent formal limits:
\begin{equation}\label{eq:cd}
C_{d,c}^P = [C_d^P]^{[M]}_{Y \to Y_c^-} = [C_d^P]^{[R]}_{B \to \infty},
\end{equation}
which characterize static stability from the physical perspective.

\subsection{Direct method for computing $Y_c$}

We now outline a direct method for computing $Y_c$ that is based on minimization of the quotient:
\begin{equation}
\frac{\int_{\Omega \setminus X} \vert\dot{\boldsymbol{\gamma}}(\mathbf{v})\vert_F}{-\int_{X} \mathbf{v}\cdot \boldsymbol{e}_z} , \label{eq:quotient}
\end{equation}
over admissible velocity fields. Note that we assume $\hat{\rho}_p > \hat{\rho}_f$ so that the particle falls in the $-\boldsymbol{e}_z$ direction, i.e.~the denominator is generally positive. The relevance of the above quotient comes from the following definition of $Y_c$:
\begin{equation}
Y_c := \sup_{\mathbf{v} \in H_\diamond} \frac{- \int_{X} \mathbf{v}\cdot \boldsymbol{e}_z}{\int_{\Omega \setminus X} \vert\dot{\boldsymbol{\gamma}}(\mathbf{v})\vert_F} . \label{eq:eigenval}
\end{equation}
We first explain (\S \ref{sec:Ycdefine}) why $Y_c$, as defined in \eqref{eq:eigenval}, coincides with our less physical definition earlier and afterwards outline the minimization method (\S \ref{sec:minimization}).

\subsubsection{Definition of $Y_c$}
\label{sec:Ycdefine}

Using the formulation [M] and following \cite{putz2010}, after some algebraic steps we arrive at the  following equivalent variational formulation, which involves minimizing over $ \mathbf{v} \in H_\diamond$ the functional {$\mathcal G_Y^\diamond(\mathbf{v})$}:
\begin{equation}\label{eq:variational}
\mathcal G_Y^\diamond(\mathbf{v})
= \begin{cases}
\displaystyle{\frac{1}{2} \int_{\Omega \setminus X}
\vert\dot{\boldsymbol{\gamma}}(\mathbf{v})\vert_F^2
+ Y \int_{\Omega \setminus X} \vert\dot{\boldsymbol{\gamma}}(\mathbf{v})\vert_F
+\int_{X} \mathbf{v}\cdot \boldsymbol{e}_z} & \text{ if } \mathbf{v} \in H_\diamond \\
+ \infty & \text{ else.}
\end{cases}
\end{equation}
where
\begin{equation}\label{eq:Hdiamond}
H_\diamond = \left\{ \mathbf{v} \in H_0^1(\Omega)^3 \ \middle |~ \ \int_\Omega v_z =0, \ \vert\dot{\boldsymbol{\gamma}}(\mathbf{v})\vert_F = 0 \text{ in } X, ~~ \mathbf{\nabla} \cdot \mathbf{v} = 0 \right\}.
\end{equation}
Let us denote by $\mathbf{v}_Y$ the minimizer of $\mathcal G_Y^\diamond $, i.e. the solution for given $Y$. This minimization is equivalent to the following variational inequality, for every $ \mathbf{v} \in H_\diamond$:
\begin{equation}
\int_{\Omega \setminus X} \dot{\boldsymbol{\gamma}}(\mathbf{v}_Y): \dot{\boldsymbol{\gamma}}(\mathbf{v}-\mathbf{v}_Y)
  + Y \int_{\Omega \setminus X} \vert\dot{\boldsymbol{\gamma}}(\mathbf{v})\vert_F
 -  Y \int_{\Omega \setminus X} \vert\dot{\boldsymbol{\gamma}}(\mathbf{v}_Y)\vert_F \geq - \int_{X} (\mathbf{v} - \mathbf{v}_Y)\cdot \boldsymbol{e}_z.
\label{eq:notEL}
\end{equation}
We use the inequality \eqref{eq:notEL} with both $\mathbf{v} = 0$ and $\mathbf{v} = 2\mathbf{v}_Y$, to obtain
\begin{equation}\begin{aligned}
\int_{\Omega} \vert\dot{\boldsymbol{\gamma}}(\mathbf{v}_Y)\vert_F^2
&=
\int_{\Omega \setminus X} \vert\dot{\boldsymbol{\gamma}}(\mathbf{v}_Y)\vert_F^2
=
-\int_{X} \mathbf{v}_Y \cdot \boldsymbol{e}_z
-Y\int_{\Omega \setminus X} \vert\dot{\boldsymbol{\gamma}}(\mathbf{v}_Y)\vert_F \nonumber\\
& \le 
\int_{\Omega \setminus X} \vert\dot{\boldsymbol{\gamma}}(\mathbf{v}_Y)\vert_F
\left[  \sup_{\mathbf{v} \in H_\diamond} \frac{-\int_{X} \mathbf{v}\cdot \boldsymbol{e}_z}{\int_{\Omega \setminus X} \vert\dot{\boldsymbol{\gamma}}(\mathbf{v})\vert_F}  - Y \right]
= (Y_c-Y) \int_{\Omega \setminus X} \vert\dot{\boldsymbol{\gamma}}(\mathbf{v}_Y)\vert_F  ,
\label{eq:zerograd}
\end{aligned}\end{equation}
Thanks to the homogeneous boundary conditions, this implies that $\mathbf{v}_Y = 0$ in $\Omega$ as soon as $Y \gs Y_c$.

This establishes that $Y \gs Y_c$ is sufficient for the static flow. Following the methods in \cite{putz2010} we can also study the convergence of $\mathbf{v}_Y $ to $0$. In particular, we find that as $Y \to Y_c^-$:
\begin{align}
\int_{\Omega \setminus X} \vert\dot{\boldsymbol{\gamma}}(\mathbf{v}_Y)\vert_F^2
   &=  \mathcal O([Y_c - Y]^2), \\
\int_{\Omega \setminus X} \vert\dot{\boldsymbol{\gamma}}(\mathbf{v}_Y)\vert_F
   &=  \mathcal O(Y_c - Y), \\
\int_{\Omega \setminus X} \vert\dot{\boldsymbol{\gamma}}(\mathbf{v}_Y)\vert_F
   & \ge \frac{1}{Y_c - Y}
\int_{\Omega \setminus X} \vert\dot{\boldsymbol{\gamma}}(\mathbf{v}_Y)\vert_F^2
   \ge 0 , \label{eq:conv3}
\end{align}
and that $\int_{\Omega \setminus X} \vert\dot{\boldsymbol{\gamma}}(\mathbf{v}_Y)\vert_F $ is decreasing with $Y$. It follows that
\[
-\int_{X} \mathbf{v}_Y\cdot \boldsymbol{e}_z \sim Y \int_{\Omega \setminus X} \vert\dot{\boldsymbol{\gamma}}(\mathbf{v}_Y)\vert_F ,~~\mbox{as}~~Y \to Y_c^- , \]
which suggests that minimization of the quotient \eqref{eq:quotient} defines the limiting flow.

However, whereas {$\mathcal G_Y^\diamond$} has a minimizer in $H_\diamond$, the quotient \eqref{eq:quotient} might not. This reflects the fact that this quotient does not take into account viscous dissipation so the transition between plugged and unplugged regions tends to be sharp, whereas functions in $H_\diamond$ do not have such discontinuities. To cope with this difficulty, we will use its standard relaxation to the larger space of functions of bounded deformation $\BD(\Omega)$. Working in this space allows us to directly look for minimizers of the quotient \eqref{eq:quotient}, allowing us to calculate the critical yield number with only one field.

A vector field $\mathbf{v} \in L^1(\Omega)$ is said to have bounded deformation if its symmetric gradient, that is its strain rate if the vector field represents a velocity, given by
$$\dot{\boldsymbol{\gamma}}(\mathbf{v}) := D\mathbf{v} + (D\mathbf{v})^T$$
is a vectorial Radon measure. In particular, for such a velocity field the strain rate can be supported on surfaces. The space $\BD(\Omega)$ is a Banach space with the norm
$$ \Vert \mathbf{v} \Vert_{\BD(\Omega)} := \Vert \mathbf{v} \Vert_{L^1} + \TD(\mathbf{v})$$
where $\TD$ denotes the total mass of the measure $\dot{\boldsymbol{\gamma}}(\mathbf{v})$ on $\Omega$ computed using the Frobenius norm $\vert \dot{\boldsymbol{\gamma}}(\mathbf{v})\vert_F$ of test functions with values in the subspace of symmetric matrices $\R^{d \times d}_{\mathrm{sym}}$, that is
\begin{equation}\label{eq:totaldef}\TD(\mathbf{v}) = \sup \left\{\int \mathbf{v} \cdot \div \mathbf{q}\ \middle \vert \  \mathbf{q} \in \mathcal C_0^1(\Omega, \R^{d \times d}_{\mathrm{sym}}),\, |\mathbf{q}(x)|_F^2 \ls 1 \text{ for all }x.\right \},\end{equation}
for which we have formally $\TD(\mathbf{v})=\int_{\Omega} \vert\dot{\boldsymbol{\gamma}}({\mathbf{v}})\vert_F.$
The space of bounded deformation function enjoys a compactness property \cite[II.(3.4)]{Tem85} as follows:
\begin{thm}
\label{thm:compact}
Let $\mathbf{v}_n$ be a sequence of functions in $\BD(\Omega)$ such that $\Vert \mathbf{v}_n \Vert_{\BD(\Omega)}$ is bounded. Then there exists a function $\mathbf{v} \in \BD(\Omega)$ such that, along a subsequence $n_k$, we have
\begin{equation}\label{eq:weakbdconv}\mathbf{v}_{n_k} \xrightarrow{L^1} \mathbf{v},\quad \gsym (\mathbf v_{n_k}) \stackrel{\ast}{\rightharpoonup} \gsym(\mathbf v),\text{ and} \quad\TD(\mathbf v)\ls \liminf \TD(\mathbf{v}_{n_k}).\end{equation}
\end{thm}
Thanks to this compactness result, we have
\begin{thm}\label{thm:relaxation}
The quotient \eqref{eq:quotient} has at least one minimizer $\hat{\mathbf{v}}$ in the space
\[\BD_\diamond := \left \{ \mathbf{v} \in \BD(\R^3) \ \middle \vert ~~\fint_{X} \mathbf{v} \cdot \boldsymbol{e}_z = -1, \; \mathbf{\nabla} \cdot \mathbf{v} = 0, \; \dot{\boldsymbol{\gamma}}(\mathbf{v}) = 0 \text{ on } X,\; \mathbf{v}(x) = 0 \text{ on a.e. } x \in \R^3 \setminus \Omega \right\},\]
where $\fint_X$ denotes the integral average on $X$. Moreover,
\begin{equation}\label{eq:relaxation}\TD(\hat{\mathbf{v}}) = Y_c.\end{equation}
\end{thm}
\begin{proof}
The existence of a minimizer is a direct consequence of the compactness and lower semicontinuity properties of Theorem \ref{thm:compact}, since all the constraints imposed are closed with respect to the convergence \eqref{eq:weakbdconv}.

To prove \eqref{eq:relaxation}, it is enough to produce a sequence $\mathbf v_k \in H_\diamond$ such that $\fint_{X} \mathbf{v}_k \cdot \boldsymbol{e}_z = -1$ and with
\begin{equation}\label{eq:tdconv}\int_{\Omega \setminus X} \vert \gsym(\mathbf{v}_k)\vert_F \to \TD(\mathbf v).\end{equation}
The usual strategy is to construct a sequence $\tilde{\mathbf v}_k \in \mathcal C^\infty_0$ through the use of a partition of unity $\psi_j$ of $\Omega \setminus X$ and convolving with a family of standard mollifiers, so that \eqref{eq:tdconv} holds. The difficulty in the present case is that we do not in general have $\div \tilde{\mathbf v}_k = 0$, so $\tilde{\mathbf v}_k \notin H_\diamond$. Since for $\psi_j \in \mathcal{C}^\infty_0$ we have $\div (\mathbf v \psi_j )= \mathbf v \cdot \nabla \psi_j$ and \cite[Thm.~II.2.2]{Tem85} $\mathbf v \in L^{d/(d-1)}$, we can only obtain
\[\div \tilde{\mathbf v}_k \to 0 \text{ strongly in }L^{d/(d-1)}.\]
However, as done in \cite[Thm.~II.3.4]{Tem85} for a slightly different situation, solving the elliptic problem
\[\begin{cases}
\Delta f_k = \div \tilde{\mathbf v}_k,& \text{on }\Omega\\
f_k = 0, &\text{ on }\partial \Omega,
\end{cases}\]
we obtain a unique $f_k \in \mathcal{C}^\infty_0$, which moreover satisfies the $L^{d/(d-1)}$ regularity estimate
\[\|f_k\|_{W^{2, d/(d-1)}} \ls C \|\div \tilde{\mathbf v}_k\|_{L^{d/(d-1)}}.\]
This allows us to define $\mathbf v_k = \tilde{\mathbf v}_k - \nabla f_k$, so that
\[\div \mathbf v_k = 0, \text{ and } \int_{\Omega \setminus X} \vert \gsym(\tilde{\mathbf{v}}_k)\vert_F - \vert \gsym(\mathbf{v}_k)\vert_F \to 0.\]
\end{proof}

Note that the constraint on the average flux through the particle $X$, i.e.
\begin{equation}
\fint_{X} \mathbf{v} \cdot \boldsymbol{e}_z = -1,
\label{eq:normalize}
\end{equation}
can be added thanks to the invariance of the quotient \eqref{eq:quotient} to scalar multiplication. We can therefore restrict the minimization to this space only, or indeed use any other scaling. The above flow rate scaling is that which is used later in the computational algorithm. Below to establish convergence we instead scale with the plastic dissipation.

\begin{thm}\label{thm:profile}
The rescaled velocity $\tilde{\mathbf{v}}_Y$:
 \begin{equation}\label{eq:velrescale}
  \tilde{\mathbf{v}}_Y := \frac{\mathbf{v}_Y}{\int_{\Omega \setminus X} \vert\dot{\boldsymbol{\gamma}}(\mathbf{v}_Y)\vert_F },
\end{equation}
converges as $Y \to Y_c^-$, to a minimizer in $\BD_\diamond$ of
\[ \frac{\int_{\Omega \setminus X} \vert\dot{\boldsymbol{\gamma}}(\mathbf{v})\vert_F}{-\int_{X} \mathbf{v}\cdot \boldsymbol{e}_z} . \]
\end{thm}

\begin{proof}
By construction $\tilde{\mathbf{v}}_Y$ has total deformation 1. Moreover, since we have homogeneous Dirichlet boundary conditions on $\partial \Omega$, we have the inequality \cite[II.(1.20)]{Tem85}
\begin{equation}\label{eq:L1bound}\|\tilde{\mathbf{v}}_Y\|_{L^{1}(\Omega)} \ls C \int_{\Omega} \vert\dot{\boldsymbol{\gamma}}(\tilde{\mathbf{v}}_Y)\vert_F = C \int_{\Omega \setminus X} \vert\dot{\boldsymbol{\gamma}}(\tilde{\mathbf{v}}_Y)\vert_F,\end{equation}
so Theorem \ref{thm:compact} implies that $\tilde{\mathbf{v}}_Y$ converge to some $\mathbf{v}_c$ in the space of $\BD$ functions (weakly in $\BD$, and strongly in $L^1$). From the mechanical energy balance (see \eqref{eq:zerograd})
\[
-\int_{X} \mathbf{v}_Y \cdot \boldsymbol{e}_z
-Y\int_{\Omega \setminus X} \vert\dot{\boldsymbol{\gamma}}(\mathbf{v}_Y)\vert_F
=
\int_{\Omega \setminus X} \vert\dot{\boldsymbol{\gamma}}(\mathbf{v}_Y)\vert_F^2
\sim 0,  \]
as $Y \to Y_c^-$; see \eqref{eq:conv3}. Thus, we conclude
\begin{equation}\label{Yc_converge}
\frac{Y\int_{\Omega \setminus X} \vert\dot{\boldsymbol{\gamma}}(\mathbf{v}_Y)\vert_F}{-\int_{X} \mathbf{v}_Y \cdot \boldsymbol{e}_z} \to 1 ~~\Rightarrow~~-\int_{X} \mathbf{v}_c \cdot  \boldsymbol{e}_z  = Y_c.
\end{equation}
Recalling that $\int_{\Omega \setminus X} \vert\dot{\boldsymbol{\gamma}}(\tilde{\mathbf{v}}_Y)\vert_F = 1$, the semi continuity of the total deformation implies $\int_{\Omega \setminus X} \vert\dot{\boldsymbol{\gamma}}({\mathbf{v}}_c)\vert_F  \ls 1$, which yields
$$ Y_c \int_{\Omega \setminus X} \vert\dot{\boldsymbol{\gamma}}({\mathbf{v}}_c)\vert_F  + \int_{X} \mathbf{v}_c \cdot \boldsymbol{e}_z   \ls 0,$$
which can be rewritten as
$$Y_c \ls \frac{-\int_{X} \mathbf{v}_c \cdot \boldsymbol{e}_z}{\int_{\Omega \setminus X} \vert\dot{\boldsymbol{\gamma}}({\mathbf{v}}_c)\vert_F}$$
meaning that $\mathbf{v}_c$ is a maximizer of
$$ \frac{-\int_{X} \mathbf{v} \cdot \boldsymbol{e}_z}{\int_{\Omega \setminus X} \vert\dot{\boldsymbol{\gamma}}({\mathbf{v}})\vert_F}.$$
\end{proof}
We {thus} see that a minimizer of the quotient \eqref{eq:quotient} can be obtained as a rescaled limit of the physical velocity solutions.

\subsubsection{Primal and dual problems in the continuous setting}
\label{sec:dual}
The direct formulation we have proposed involves minimizing the functional $ \int_{\Omega \setminus X} | \dot{\bs{\gamma}} (\mathbf{ v})|_F$, which is not a strictly convex or smooth functional. Our approach for approximating the minimizers will be based on exploring convex duality for this problem, in which one introduces stress variables.

Generically, for a minimization problem of the type
\[ \inf_{v \in V} F(v)+G(Av), \]
where $F$ and $G$ are convex and $A:V \to Y$ is a linear operator between two Banach spaces $V$ and $Y$, we have the weak duality inequality between this problem and its Fenchel dual \cite[Ch. III, Theorem 4.2 and Remark 4.2]{EkeTem76}

\begin{equation} \sup_{y^\ast \in Y^\ast}  -F^\ast( A^\ast y^\ast) - G^\ast (-y^\ast) \label{eq:dualabstract} \ls \inf_{v \in V} F(v)+G(Av).\end{equation}
Here $F^\ast$ and $G^\ast$ are the Fenchel conjugates of $F$ and $G$, defined by
\[ F^*(v^*) := \sup_{v\in V} \{ \scal{v^\ast}{v} - F(v) \}.\]

This inequality means that for a given velocity field $v$, if we can find a dual field $y$ in which this inequality is an equality, this would certify that the velocity field is a minimizer of \eqref{eq:quotient}. This situation is referred to as strong duality.

To explicitly compute this dual problem, we need to define the symbols above and in particular take care of the divergence constraint and boundary conditions. A constraint $v\in K$ can be treated as part of the objective functional, by adding the characteristic function
\[ \rchi_K(v) := \begin{cases} 0 & \text{ if }v \in K \\ +\infty &  \text{ otherwise}. \end{cases}\]
On the other hand, we could also use a Lagrange multiplier to take care of one or more constraints, leading to a saddle point formulation. In what follows, we derive a  formulation based on the deviatoric stress, and the saddle point formulation that will be used in the numerical computations, where incompressibility is treated by a Lagrange multiplier.

\subsection{The dual problem in terms of stress}
We start noting that our problem has the same minimal energy value as minimizing
\begin{equation}\label{eq:H1gsym}\int_{\Omega \setminus X} | \dot{\bs{\gamma}} (\mathbf{ v})|_F\end{equation}
over $\mathbf{v} \in H^1(\Omega \setminus X)^3$ with boundary conditions $\mathbf{v} = \mathbf{v}_0$ on $\partial (\Omega \setminus X)$ and $\div \mathbf{v} = 0$, where $\mathbf{v}_0$ is {a} divergence-free field in $H^1(\Omega \setminus X)^3$ such that its trace on $\partial (\Omega \setminus X)$ equals the desired boundary conditions, but otherwise arbitrary. For regular enough boundary values, as in the cases of interest for us, such a function can be found by solving a {linear} Stokes problem on $\Omega \setminus X$. As discussed in previous sections, it is generally the case that the minimization of \eqref{eq:H1gsym} in subspaces of $H^1(\Omega \setminus X)^3$ does not admit solutions, since the space $H^1$ does not allow the kind of discontinuities that the limiting velocity field presents. However, in the following discussion only the minimal energy values are relevant.

Since the linear subspace $H^1_{\mathrm{div}} \subset H^1(\Omega \setminus X)^3$ of functions with zero divergence is closed, it is itself a Banach space and we can then write the primal problem above as
\begin{equation}\label{eq:primal} \inf_{\mathbf{v} \in H^1_{\mathrm{div}}}  G(\gsym (\mathbf{v})) +  F (\mathbf{v})\end{equation}
with $$ G(\mathbf{q}) = \int |\mathbf{q}|_F \text{, and }  F(\mathbf{v}) = \rchi_{H^1_0 + \mathbf{v}_0}(\mathbf{v}).$$
When defined in $H^1_{\mathrm{div}}$, $\gsym$ maps onto the space of symmetric traceless matrices with $L^2$ coefficients, which we denote by $L^2_\text{sym,tr}$. We denote then the dual variable by $\bs \zeta \in L^2_\text{sym,tr}$. With the definition of the conjugate functions we get
\begin{align*}
G^\ast(\bs \zeta) & = \sup_{\mathbf{q}  \in L^2_{\text{sym,tr}}}  \int \bs \zeta:\mathbf{q} -  \int |\bs \zeta|_{F} = \sup_{\mathbf{q}  \in L^2_{\text{sym,tr}}} \int \left(\bs \zeta : \big(\mathbf{q} - \frac{\bs \zeta}{|\bs \zeta|_{F}}\big)\right) \\
 & = \left \{ \begin{matrix} +\infty &\text{ if }|\bs \zeta |_{F} > 1 \\ 0 &\text{ otherwise,} \end{matrix} \right.
\end{align*}
and for {an auxiliary velocity-type dual variable} $\bs \xi \in (H^1_{\mathrm{div}})^\ast$, assuming that $\div \mathbf v_0 = 0$ we get
\begin{align*} F^\ast(\bs \xi) &= \sup_{\mathbf{v} \in (H^1_0 + \mathbf{v}_0) \cap H^1_{\mathrm{div}}} \int \bs \xi \cdot \mathbf{v} = \sup_{\mathbf{v} \in (H^1_0 + \mathbf{v}_0) \cap H^1_{\mathrm{div}}} \int \bs \xi \cdot(\mathbf{v}-\mathbf{v}_0) + \int \bs \xi \cdot \mathbf{v}_0 \\ &= \sup_{\mathbf{v} \in H^1_0 \cap H^1_{\mathrm{div}}} \int \bs \xi \cdot \mathbf{v} + \int \bs \xi \cdot \mathbf{v}_0 = \left \{ \begin{matrix} \int \bs \xi \cdot \mathbf{v}_0 &\text{ if } \int \bs \xi \cdot \mathbf{v} = 0 \text{ for all } \mathbf{v} \in H^1_0 \cap H^1_{\mathrm{div}}\\ +\infty &\text{ otherwise.}\end{matrix} \right.
\end{align*}
Now, the dual problem
$$ \sup_{\bs{\zeta} \in L^2_{\text{sym,tr}}} -  G^\ast( -\bs \zeta ) -  F^\ast ( \gsym^\ast \bs \zeta ) $$
{is formulated just in terms of $\zeta$, since its standard form \eqref{eq:dualabstract}} involves only $\bs \xi$ of the form $\gsym^\ast \bs \zeta = \div \bs \zeta$, where $\div \bs \zeta$ is understood as a vector distribution in $[\mathcal{C}_0^\infty(\Omega, \R^d)]^\ast$. Therefore, using that $\tr \bs \zeta = 0$ we have
\[F^\ast(\div \bs \zeta) = \left \{ \begin{matrix} \int \bs \zeta : \gsym \mathbf{v}_0 &\text{ if } \int \div \bs \zeta \cdot \mathbf{v} = 0 \text{ for all } \mathbf{v} \in H^1_0\\ +\infty &\text{ otherwise.}\end{matrix} \right.\]
Using the de Rham theorem for distributions \cite[Proposition 1.1]{Tem77}, the condition $\int \div \bs \zeta \cdot \mathbf{v} = 0 \text{ for all } \mathbf{v} \in H^1_0$ means that there exists a scalar function $p \in L^2(\Omega)$ such that
\begin{equation}\label{eq:L2pressure}\div \bs \zeta = \nabla p, \text{ or equivalently, }\curl \div \bs \zeta = 0.\end{equation}
Gathering up the above we have derived the following dual problem in terms of the deviatoric part of the stress:
\begin{equation}\sup_{\substack{\bs \zeta \in L^2_{\text{sym,tr}} \\ \curl \div \bs \zeta = 0 \\ |\bs \zeta|\ls 1}} \int_{\Omega \setminus X} \gsym \mathbf v_0 : \bs \zeta.\label{eq:dualexpl} \end{equation}
Since $G$ is continuous on $L^2_\text{sym,tr}$, this formulation has the advantage of enjoying strong duality \cite[Ch. III, Theorem 4.2 and Remark 4.2]{EkeTem76}, implying the existence of an optimal $\bs \zeta$ and that one can obtain the critical yield number through this dual problem. However, it would be impractical to numerically implement the second order differential constraint $\curl \div \bs \zeta = 0$. We can instead add the hydrostatic contribution found in \eqref{eq:L2pressure} to the stress and formulate the dual problem in terms of the variable $\bs \zeta + p \Id$, to reformulate the dual problem as
\begin{equation}\sup_{\substack{\bs \zeta \in L^2_{\text{sym}} \\ \div \bs \zeta = 0 \\ |\bs \zeta - \frac 13 \tr \bs \zeta \Id| \ls 1}} \int_{\Omega \setminus X} \gsym \mathbf{v}_0 : \bs \zeta,\label{eq:dual} \end{equation}
in which we recognize the characteristics of the complete stress, with no external forces and satisfying the von Mises criterion for its deviatoric part, as should be the case for a Bingham fluid. Let us note that \eqref{eq:dual} could have also been derived directly by considering the whole space $H^1(\Omega \setminus X)^3$, the corresponding linear operator $\gsym$ and the functions
$$\tilde G(\mathbf{q}) = \int |\mathbf{q}|_F + \rchi_{\tr \cdot = 0} (\mathbf{q}) \text{, and }  F(\mathbf{v}) = \rchi_{H^1_0 + \mathbf{v}_0}(\mathbf{v}).$$
However in this formulation $\tilde G$ is not continuous, which would have made proving strong duality less straightforward. In any case one could still use the Attouch-Brezis qualification condition \cite{AttBre86}, which requires the sum of the subspaces $H^1_0$ and $H^1_{\mathrm{div}}$ to be closed. This closedness condition can then be seen to be equivalent to finding $p \in L^2$ as in \eqref{eq:L2pressure} for each dual variable in \eqref{eq:dualexpl}.

Let us further remark, that while $Y_c = \TD(\mathbf{v})$ for some optimal velocity field $\mathbf{v} \in \BD_\diamond$, and also $Y_c = \int_{\Omega \setminus X} \gsym(\mathbf{v}_0) : \bs \zeta$ for some $\bs \zeta \in L^2_{\mathrm{sym}}$ optimal in \eqref{eq:dual}, the energy dissipation rate $\int_{\Omega \setminus X} \gsym(\mathbf{v}) : \bs \zeta$
is \emph{not defined} for these two limit formulations together, since the fields considered are not regular enough.

\subsection{Continuous primal-dual formulation}
Since, as mentioned above, treating $\curl \div \bs \zeta$ as a second order constraint would not be feasible numerically, we now turn to the variant of \eqref{eq:dual} which will be discretized, with a multiplier for the incompressibility constraint. The linear operator we now consider is
$$(\gsym, \div):H^1(\Omega)^3 \to L^2_\text{sym} \times L^2,$$
leading to an additional dual scalar field. The function $F$ of the previous section is maintained, while for the second term we use instead
$$\tilde G(\mathbf{q}, \vartheta) = \int |\mathbf{q}|_F + \rchi_{\{0\}} (\vartheta).$$
Using the observations made at the start of the previous section and \cite[Ch. III, Remark 3.2 and Equations (4.14),(4.15)]{EkeTem76}, we also have strong duality for the resulting problem, so that
\begin{equation}\begin{aligned}
\min_{\mathbf{v} \in \BD_\diamond} \TD(\mathbf{v}) &= \inf_{\mathbf{v} \in H^1} \rchi_{\{0\}}(\div \mathbf{v})+\int |\gsym (\mathbf{v})|_F + \rchi_{H^1_0 + \mathbf{v}_0}(\mathbf{v})\\
&= \inf_{\mathbf{v} \in H^1} \sup_{\substack{\bs \zeta \in L^2_{\text{sym}} \\ p \in L^2}} \int \gsym (\mathbf{v}) : \bs \zeta + \int \div(\mathbf{v})\, p + \rchi_{H^1_0 + \mathbf{v}_0}(\mathbf{v}) - \rchi_{|\cdot|_F \ls 1}(\bs \zeta)\\
& = \sup_{\substack{\bs \zeta \in L^2_{\text{sym}} \\ p \in L^2}} \inf_{\mathbf{v} \in H^1} \int \gsym (\mathbf{v}) : \bs \zeta + \int \div(\mathbf{v})\, p + \rchi_{H^1_0 + \mathbf{v}_0}(\mathbf{v}) - \rchi_{|\cdot|_F \ls 1}(\bs \zeta),
\end{aligned}\label{eq:contprimaldual}\end{equation}
which is the problem that we will discretize for our numerical method. We have maintained the constraints as characteristic functions, since these will be treated by projection in the algorithm. We will prove below that the minimal values of the discrete minimization converge to those of \eqref{eq:contprimaldual}, recovering the critical yield stress and corresponding velocity field.

\section{Computational methods}
\label{sec:computing}

This paper is focused at proposing a general method for computing $Y_c$ \emph{directly}, based on the primal-dual formulation of the previous section. We develop and describe this method below in \S \ref{sec:minimization}. In order to benchmark this method we test against computations carried out on axisymmetric particles, which are computed using the full viscoplastic problem. We outline this method in \S \ref{sec:axisymmetric_calcs}.

\subsection{Direct computation of the yield limit}\label{sec:minimization}

We aim to discretize the problem of minimization of $\int |\dot{ \bs \gamma}(\bs{\cdot})|_F$ in $\BD_\diamond$ under the constraint
$$\mathbf{v}_0 =  -\mathbf{e}_z \quad \text{on } X,$$
to subsequently compute the corresponding discrete minimizers. Let us note that in our situation the functional does not have any term or constraint which is strongly convex or differentiable. This fact, as opposed to the situation when computing non-limiting flows, makes most types of acceleration schemes (see \cite{TreMoyPri16, TreRouFriWac18} in the context of viscoplastic flows) non-applicable. This can be seen as the price to pay for a direct method in which only one limiting field is computed, instead of a fixed approximation by computing Stokes flows. Given this situation, we use the non-accelerated version of the primal-dual proximal splitting algorithm of \cite{ChaPoc11} applied to a discretization of the primal-dual formulation \eqref{eq:contprimaldual}. For it we have chosen (in their notation) the dual step size as $\tau = 0.1$, the primal step size as $\sigma = 1/ \tau / L^2$, where $L^2$ is an estimation of the norm of the discrete operator $(\nabla, \div)$, and the relaxation parameter as $\theta = 1$.

For simplicity we will consider particles $X$ with 3 orthogonal planes of symmetry. Under this assumption, we expect that no rotational torque is applied on the particle and there is a steady settling flow, (see \cite[Sec.~4]{Wei72} for a proof in the case of Stokes flow), and if gravity is oriented along a symmetry plane then there is no \emph{drift}, i.e.~the particle settles vertically. The minimizers represent the yield limit, as discussed. However, we may consider 3-dimensional particles without requiring cylindrical symmetry and hence assess the method. The symmetry allows us then to perform numerical computations just on an octant, which gives some computational economy.

Since the fields we want to compute are in general discontinuous, we have opted for a finite difference discretization on a rectangular grid $G=\{1,\ldots, n_t\}\times\{1,\ldots n_t\}\times\{1,\ldots n_z\}$, where due to the symmetry we have chosen identical sizes $n_t$ in the two directions orthogonal to gravity. Choosing the grid sizes $(n_t, n_t, n_z)$ is not a trivial matter: while we expect the velocities to be supported in a bounded set and therefore completely vanish before they reach $\partial \Omega$, we don't know a priori how big the unyielded set is. A computational domain that is too small to fit the unyielded region will modify the resulting flow through the  conditions imposed on its boundary, and one that is too big will waste computational effort and limit the attainable precision. In practice, we have aimed to use only one resolution parameter $n = \max(n_t, n_z)$ for easier comparison with similar computational costs across different shapes, and denoted the resulting grid by $G^n$. The other parameter is then computed using a heuristic based on the aspect ratio of the particle, which although it has no theoretical guarantees, has worked in all of our experiments. The idea of our heuristic is to use the results of \cite{ChaFri17} where the slipline theory of perfect plasticity \cite[Sec.~5.1.1]{Lub08} is used to find unyielded envelopes of particles in the limiting 2D flow. In particular, there it is obtained that the unyielded envelope of a square falling with a diagonal aligned with gravity should be given by a pair of circles with the sides of the square as radii. The heuristic is then to enclose the particle in such a square, and size the grid slightly larger than the computed unyielded region.

We can now write the discrete primal-dual formulation. The constraint $\int_{\Omega} \mathbf{v}_0 \cdot \mathbf{e}_z = 0$ is seen as a Lagrange multiplier, as is the incompressibility, whereas the conditions $\mathbf{v}_0 = 0$ on $\partial \Omega$ and $\mathbf{v}_0 = -\mathbf{e}_z$ on the particle are encoded as indicatrix functions, such that the backward step is a projection. Finally, the functional to optimize on is
\begin{equation}\label{eq:discmin}\begin{gathered}
\min_{\mathbf{v} \in S^3}  \rchi_{K^n}(\mathbf{v}) + \rchi_{\{\div^n_C = 0\}}( \mathbf v) + \sum_{(i,j,k) \in G^n} |\gsym^n(\mathbf{v})^{ijk}|_F \\
=\min_{\mathbf{v} \in S^3} \max_{\substack{\mathbf{q} \in S^{12} \\ \mkern-12mu \omega \in \mathbb R \\ \mkern-12mu p \in S}} \rchi_{K^n}(\mathbf{v}) + \sum_{(i,j,k) \in G^n} \gsym^n(\mathbf{v})^{ijk} : \mathbf{q}^{ijk} - \rchi_{\{|\cdot|_{\infty} \leq 1\}}(\mathbf{q}^{ijk})-\omega\,\mathbf{v}^{ijk}_3-p^{ijk} \left( \div^n_C \mathbf{v}\right)^{ijk} \\ =
\max_{\substack{\mathbf{q} \in S^{12} \\ \mkern-12mu \omega \in \mathbb R \\ \mkern-12mu p \in S}} \min_{\mathbf{v} \in S^3} \rchi_{K^n}(\mathbf{v}) + \sum_{(i,j,k) \in G^n} \gsym^n(\mathbf{v})^{ijk} : \mathbf{q}^{ijk} - \rchi_{\{|\cdot|_{\infty} \leq 1\}}(\mathbf{q}^{ijk})-\omega\,\mathbf{v}^{ijk}_3-p^{ijk} \left( \div^n_C \mathbf{v}\right)^{ijk}.
\end{gathered}\end{equation}
Notice that for this discrete problem we do have strong duality, so that we can indeed write the equality above. Since the domains of the functions involved are the linear subspaces $K^n$ and $\{\mathbf v \in S^3 | \div^n_C \mathbf{v} = 0 \}$, and the relative interior of a subspace is the subspace itself, it is enough \cite[Cor.~31.2.1]{Roc70} to notice that there is at least one discrete velocity field $\mathbf v_0$ with $\div^n_C \mathbf{v}_0=0$ which also satisfies the boundary conditions. Such a $\mathbf v_0$ can be found, for example, by using the projection of Lemma \ref{lem:discdivproj} in the appendix.

In \eqref{eq:discmin}, $S=\R^{n_t n_t n_z}$ denotes the space of real-valued scalar discrete functions defined on the nodes of the grid $G^n$, with corresponding vector fields in $S^3$ and pairs of symmetric tensors in $S^{12}$, with superscripts $\mathbf{v}^{ijk} \in \R^3$ denoting values at gridpoints. The discrete constraint set $K^n$ (using only one grid size parameter $n$ as described for the construction of $G^n$) is defined as
\[K^n=\left\{ \mathbf{v} \in S^3 \,\middle \vert\, \mathbf{v} = 0 \text{ on } G^n \setminus \Omega^n, \text{ and } \mathbf{v} = {-}(0,0,1) \text{ on }X^n \cap G^n\right\}\]
where $\Omega^n$ and $X^n$ denote the discretized domain and particle, to be defined in \eqref{eq:discretedomain} and \eqref{eq:discretesolid} below. The discrete divergence $\div^n_C$ is defined with centered differences, so for $\mathbf{v}^{ijk} = \left[ v_1^{ijk}, v_2^{ijk}, v_3^{ijk}\right]$ we have
\[\left(\div^n_C \mathbf{v}\right)^{ijk} = n\left[\left( v_1^{i+1,j,k}-v_1^{i-1,j,k}\right)+\left( v_2^{i,j+1,k}-v_1^{i,j-1,k}\right)+\left( v_3^{i,j,k+1}-v_3^{i,j,k-1}\right)\right],\]
where the factor $n$ is the inverse of the discretization step $1/n$. We have used a multiplier $\omega$ for the constraint $\sum \mathbf v_3 = 0$ which is redundant, since it is implied by $\div^n_C \mathbf v = 0$ and $\mathbf v \in K^n$, but was observed to slightly enhance the convergence of the method. Finally, $\dot{ \bs \gamma}(\mathbf{v})^{ijk}$ stands for a discrete symmetric gradient suitable for capturing the total deformation of jumps, whose choice we now discuss.

\subsubsection{Discretization}
We discretize the problem using a generalization of the ``upwind'' scheme of \cite{ChaLevLuc11} which has the advantage of carrying a high degree of isotropy. The discrete velocity is denoted by $\mathbf{v}^{ijk} = \left[ v_1^{ijk}, v_2^{ijk}, v_3^{ijk}\right].$ The expressions for the discrete gradient and the corresponding divergence parallel those of \cite{ChaLevLuc11}, and we define for $\alpha = 1,2,3$
\begin{align*}
& \left( (\nabla \mathbf{v})^{ijk}_{\alpha,1,+}, \ (\nabla \mathbf{v})^{ijk}_{\alpha,1,-}, \ (\nabla \mathbf{v})^{ijk}_{\alpha,2,+}, \ (\nabla \mathbf{v})^{ijk}_{\alpha,2,-}, \ (\nabla \mathbf{v})^{ijk}_{\alpha,3,+}, \ (\nabla \mathbf{v})^{ijk}_{\alpha,3,-} \right) :=\\
&\quad n\,\left(v_\alpha^{i+1,j,k}-v_\alpha^{i,j,k}, \  v_\alpha^{i-1,j,k} - v_\alpha^{i,j,k},\ v_\alpha^{i,j+1,k} - v_\alpha^{i,j,k},\ v_\alpha^{i,j-1,k} - v_\alpha^{i,j,k}, \ v_\alpha^{i,j,k+1} - v_\alpha^{i,j,k}, \ v_\alpha^{i,j,k-1} - v_\alpha^{i,j,k}  \right),
 \end{align*}
where we remark that we have both forward and backward differences, doubling the number of components needed, bringing the total up to $18$ in three dimensions. From which can define the components of the discrete symmetrized gradient $\gsym(\mathbf{v})^{ijk}$ for $\alpha, \beta = 1,2,3$ by
\[\gsym(\mathbf{v})^{ijk}_{\alpha, \beta, \pm} = (\nabla \mathbf{v})^{ijk}_{\alpha,\beta,\pm} + (\nabla \mathbf{v})^{ijk}_{\beta,\alpha,\pm},\]
which implies that at each grid point $(i,j,k)$ the values of the discrete symmetrized gradient and of its multiplier $\gsym(\mathbf{v})^{ijk}, \mathbf{q}^{ijk} \in \R^{12}$.

It is important to use a discretization that takes into account derivatives in all directions as equally as possible, since we aim to resolve sharp geometric interfaces of the discontinuous flows. Illustrations of the directional behaviour of different finite difference schemes when finding interfaces in the anti-plane case can be found in \cite{IglMerSch18b}, where the discretization of \cite{ChaLevLuc11} is found to be particularly symmetric, as expected by its construction. We remark in any case that when only forward differences are used, The geometry of the interfaces is distorted according to their orientation. Using centered differences is also not adequate, since the centered difference operator has a nontrivial kernel and our solutions are constant in large parts of the domain.

\subsubsection{Discrete boundary conditions and convergence}\label{sec:convergence}

We now define the discrete domain $\Omega^n$ and particle $X^n$, and formulate a convergence result to demonstrate that the chosen discretization and penalization scheme is consistent and correctly accounts for the boundary conditions in the limit. The definitions and analysis follow the same lines as for the anti-plane case as presented in \cite{IglMerSch18b}. In this section, for simplicity, we assume that $n_t=n_z=n$ and $\Omega \Subset (0,1)^3$. We introduce the (continuous) rectangle
$$R_{ijk}^n := \frac 1n \left( i-\frac 12, i+ \frac 12 \right) \times \left( j-\frac 12, j+ \frac 12 \right)\times \left( k-\frac 12, k+ \frac 12 \right).$$
First, we need to decide which constraint to use in the discrete setting. We denote by $$R_{ijk}^n - B\left(\frac 1n\right) := \left\{x \in R_{ijk}^n \, \middle\vert\,  d\big(x,\partial R_{ijk}^n\big) > \frac 1n \right\}.$$ Our choice is to take
\begin{equation}\label{eq:discretesolid}X^n := \bigcup_{R_{ijk}^n \subset X - B(\frac 1n)} R_{ijk}^n\end{equation}
whereas
\begin{equation}\label{eq:discretedomain}\Omega^n := [0,1]^3 \setminus \left( \bigcup_{R_{ijk}^n \subset ([0,1]^3 \setminus \Omega) - B(\frac 1n)} R_{ijk}^n \right),\end{equation}
such that the discrete constraints are less restrictive than the continuous ones (so that the derivatives at the object and particle boundaries are taken into account) and
\begin{equation}\label{eq:discompact}\overline{X^n} \Subset X, \quad  \overline{[0,1]^3 \setminus \Omega^n} \Subset [0,1]^3 \setminus \Omega.\end{equation}
We define $\TD^n$ by analogy with $\mathrm{TV}^n$ in \cite{ChaLevLuc11}, when the function is piecewise constant on the $R_{ijk}^n$ and $+\infty$ otherwise, which leads to
$$\TD^n := \frac{1}{n^3} \sum_{(i,j,k) \in G^n} |\dot{ \bs \gamma}(\mathbf{v})^{ijk} \vee 0|_F$$
with $\dot{ \bs \gamma}(\mathbf{v})^{ijk} \vee 0$ denoting the positive components of the discrete shear rate $\dot{ \bs \gamma}(\mathbf{v})^{ijk}$.

This discretization is consistent with the continuous definitions, as reflected in the following result, proved in \ref{sec:conv}.
\begin{thm}\label{thm:conv}The discretization used converges in the sense of $\Gamma$-convergence with respect to the $L^1$ topology, that is
$$ \TD^n + \rchi_{C^n} + \rchi_{\{\div^n_C = 0\}} \xrightarrow{\Gamma-L^1} \TD + \rchi_C+ \rchi_{\{\div = 0\}}$$
where $$C^n := \left\{\mathbf{v}={-}(0,0,1) \text{ on } X^n,\, \mathbf{v}=0 \text{ on } G^n \setminus \Omega^n\right\}$$ and $$C := \left\{\mathbf{v}={-}(0,0,1) \text{ on } X,\, \mathbf{v}=0 \text{ on } [0,1]^3 \setminus \Omega)\right\}.$$
\end{thm}

This result implies that sequences of exact discrete minimizers converge in the $L^1$ topology to minimizers in $\BD$ of the continuous problem (see Corollary \ref{cor:convmin} in the appendix). However, the method used does not allow us to prove any convergence rates. { We will come back to this point later in Fig.~\ref{fig:schematic_grid}.}

\subsection{Axisymmetric computations using the full viscoplastic flow}
\label{sec:axisymmetric_calcs}

The basic methods that we have used for the full viscoplastic flow computations use a finite element (FE) discretisation of the Stokes equations and solution via the augmented Lagrangian method. The FE discretisation is coupled with adaptive meshing, as developed by Saramito and co-workers \cite{roquet2003}. Modification for problem [M] follows the method in \cite{putz2010}. The implementation and mesh adaptation is conducted in C++ open-source environment---FreeFEM++ \cite{MR3043640}. In the previous works, we have benchmarked this implementation \cite{chaparian2017yield,chaparian2017cloaking,chaparian2018inline,chaparian2019porous,chaparian2019adaptive} where we have performed extensive computations, more specifically 2D planar flows around particles of different shapes and orientations.

For the yield limit problem in \cite{chaparian2017yield}, we have approximated the yield limit either by solving [M] or [R] problems (see \S \ref{sec:problem}), and iteratively increasing $Y$ or $B$, respectively. In the former case $Y_c$ is defined as the value of $Y$ at which the velocity becomes zero. When using problem [R] we have found that an acceptable approximation to $Y_c$ (via the limiting plastic drag coefficient mapping) is typically achieved for $B \gtrsim 10^4$.

In this paper we have adapted the method further, namely to computing axisymmetric flows. The computational framework is very similar to the 2D planar flows with minor modifications, as we now have also a hoop stress. Within the iterative loop of the Uzawa algorithm, at each iteration we solve a (linear) axisymmetric Stokes flow problem, using the existing FreeFEM++ implementation. It is necessary to modify the other steps of the algorithm, by introducing extra variables for the hoop stress (Lagrange multiplier) and the corresponding strain rate. For the results in this paper we have used problem [R] and continuously increased the Bingham number ($B \to \infty$) until the asymptotic value is achieved for the plastic drag coefficient ($C_d^P \to C_{d,c}^P$) or equivalently an asymptotic value of the plastic dissipation, $j(\boldsymbol{u}^*)$. We illustrate this procedure in Fig.~\ref{fig:Mesh}.

\begin{figure}[ht]
\centering
\includegraphics[width=0.9\textwidth]{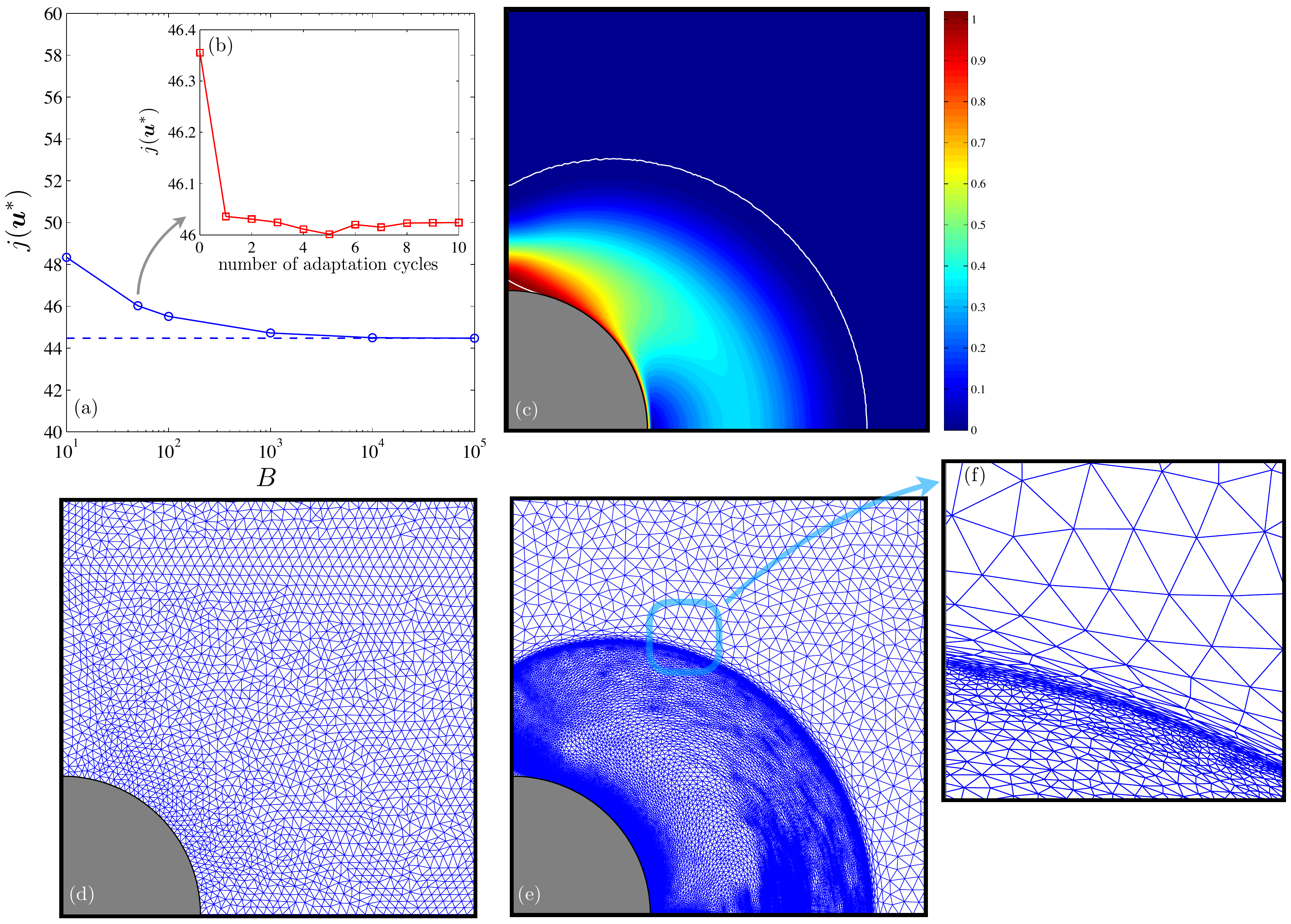}
\caption{[R] problem performance for computing the yield limit of a sphere: (a) plastic dissipation in the [R] problem, $j(\boldsymbol{u}^*)$ versus $B$; (b) $j(\boldsymbol{u}^*)$ versus the number of mesh adaptation cycles; (c) contour of $\vert \boldsymbol{u}^* \vert$; (d) initial mesh; (e) mesh after 10 cycles of adaptation; (f) zoom of the blue window in panel (e). Panels (b-f) correspond to the case $B=50$.}
\label{fig:Mesh}
\end{figure}

Panel (a) shows how $j(\boldsymbol{u}^*)$ asymptotes to its limiting value for a sample computation around a sphere. As with the planar flows, $B \gtrsim 10^4$ produces acceptable results. For the case $B=50$, panel (b) shows the convergence against the cycles of adaptation to collect the computational data. The magnitude of the flow velocity is contoured in panel (c), the initial mesh in panel (d) and the adapted mesh after 10 cycles in panel (e), with a zoom on the yield surface resolution in the panel (f).

\section{Results}
\label{sec:results}

We first present comparative results between the full viscoplastic flow computations and those of the direct method. These comparisons are confined to axisymmetric particles. Having established that the direct method is effective, we present new results on non-axisymmetric particles, to act as benchmarks for future computation.

\subsection{Axisymmetric particles}

Creeping flow of yield-stress fluid around a sphere has been studied many times and so we give a brief historical review. It seems that \cite{volarovich1953theory} were the first who found that if a particle moves within a viscoplastic medium then it should do so in a bounded subset of yielded fluid that surrounds the particle. Generally, the stress decay with distance from an isolated particle and hence will fall below the yield stress if we go sufficiently far from the particle. Following that, Andres \cite{andres1960equilibrium} defined the notion of a `sphere of influence' to refer to the yielded fluid around a sphere which is moving within a yield-stress fluid. A series of experiments confirmed that the particle does indeed move inside an unyielded envelope; see \cite{brookes1969drag,valentik1965terminal}. It was also noted from symmetry arguments that some part of rigid material might attach to the particle surface. Attempts were made to predict the shape of the rigid zone attached to the particle and the shape of the outer yield surface, either theoretically or experimentally, e.g.~\cite{ansley1967motion,yoshioka1971creeping}.

\begin{figure}
\centering
\includegraphics[width=0.7\textwidth]{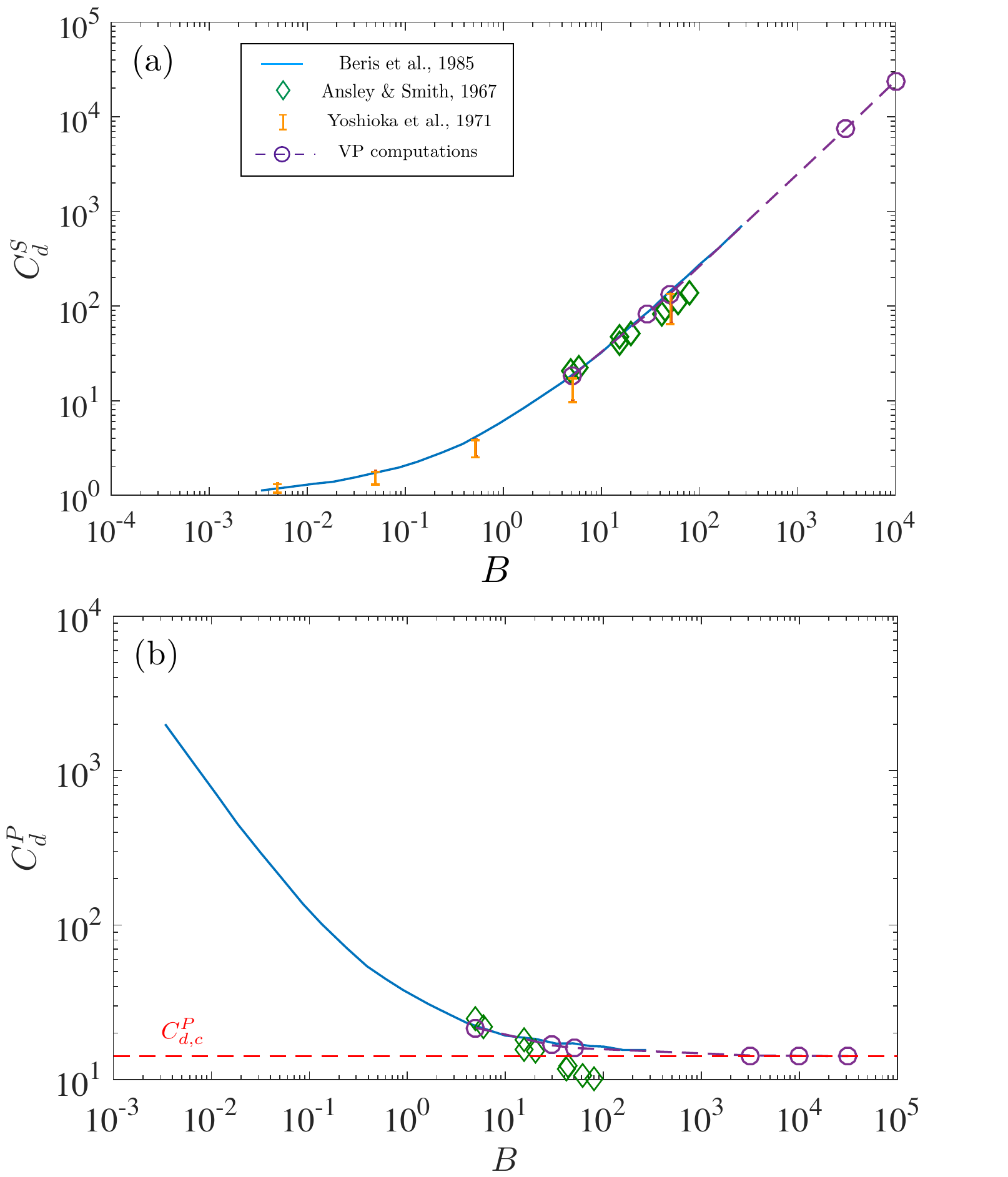}
\caption{Drag coefficients for a sphere: (a) Stokes drag,~(b) Plastic drag. The continuous blue line is the data from \cite{beris1985} and diamonds from \cite{ansley1967motion}. Orange symbols mark the upper and lower bounds calculated by \cite{yoshioka1971creeping} and the dashed line-circles show the present axisymmetric viscoplastic computations.}
\label{fig:SphereDrag}
\end{figure}

Regarding experimental studies, although there is qualitative similarity for moving particles and even some reasonable comparisons of drag coefficients \cite{tabuteau2007drag}, the critical limit is not easily resolved experimentally. This is partly due to the ambiguous nature of yield stress measurements and partly due to other physical complexities at low shear. On the other hand, the limit exists practically speaking as yield stress fluids do manage to suspend small particles indefinitely.

Theoretical questions regarding the yield surface shapes were resolved by Beris et al.~\cite{beris1985}, who performed a rigorous numerical study to find the shapes of these envelopes. Additionally the limiting flow was evaluated accurately. Beyond \cite{beris1985} there have been relatively few studies of axisymmetric particle motion in yield stress fluids that deal explicitly with the yield limit.

\subsection{Limiting viscoplastic flow}

Here we use length-scale $\hat{L} = \hat{R}$, and define the Stokes drag coefficient as:
\begin{equation}\label{eq:Stokesdrag}
C_d^S = \frac{\hat{F}^*}{6 \pi \hat{\mu} \hat{U}_p \hat{R}} = \frac{1}{6 \pi} \int_{\partial X} \boldsymbol{e}_z \cdot \boldsymbol{\sigma}^* \cdot \boldsymbol{n} ~\text{d}S,
\end{equation}
which  represents the ratio of the drag force experienced by a sphere in Bingham fluid ($\hat{F}^*$) compared to that of a viscous fluid $(6 \pi \hat{\mu} \hat{U}_p \hat{R})$. Beris et al. \cite{beris1985} have used a similar idea, however it seems that there is a typo in their paper, since it differs by a factor of $2 \pi$ from expression \eqref{eq:Stokesdrag}. The Stokes drag coefficient is of a less interest in context of yield-stress fluids and especially in the problem of the yield limit, since when $B \to \infty$, also $C_d^S \to \infty$. However, it is of use to benchmark our results with those of \cite{beris1985} in Fig.~\ref{fig:SphereDrag}a. For a sphere, the plastic drag coefficient is:
\begin{equation}
C_d^P = \frac{6}{B}~ C_d^S,
\end{equation}
which can be used to convert data. Figure \ref{fig:SphereDrag}b plots the plastic drag coefficient for a sphere, from which as $B \to \infty$, we can extract the asymptotic values of the plastic drag coefficient, $C_{d,c}^P$. We compute $C_{d,c}^P \sim 14.18$, which is equivalent to $Y_c = 0.094$. Beris et al.~\cite{beris1985} reported $Y_c = 0.0953$. The small difference may find its root in the different numerical methods.

\begin{figure}
\centering
\includegraphics[width=0.5\textwidth]{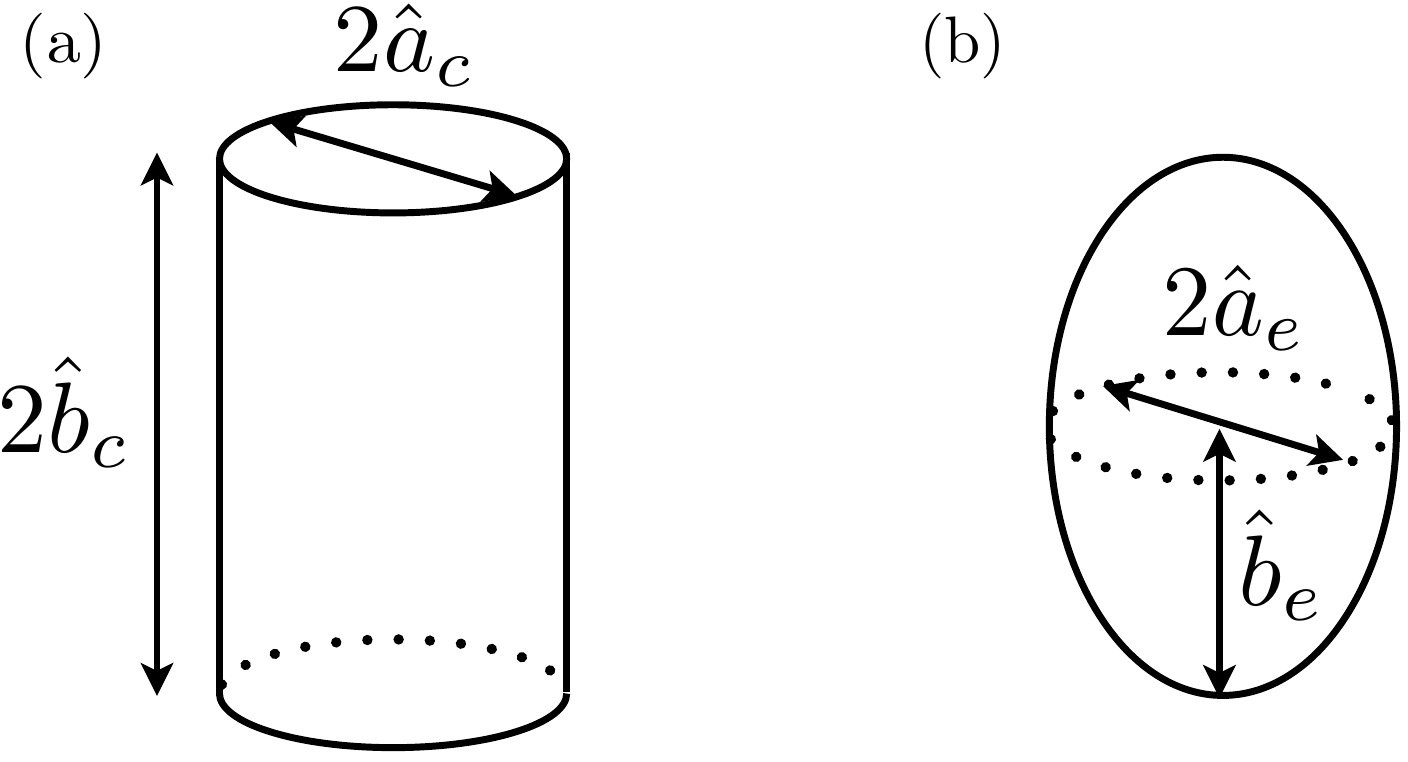}
\caption{Schematic defining the geometry of: (a) a cylindrical particle; (b) an ellipsoidal particle.}
\label{fig:schematic}
\end{figure}

Figure \ref{fig:SphereDrag}b establishes the effectiveness of the full visco-plastic code in computing the critical limits. We have also used the same code to compute the critical limit for 2 other axisymmetric particles: cylinders and ellipsoids, as illustrated in Fig.~\ref{fig:schematic}. These particles are made dimensionless using a volumetrically defined radius, i.e.~$\hat{L} = \left[ (3 \hat{V})/(4 \pi) \right]^{\frac{1}{3}}$, so that for a sphere $\hat{L} = \hat{R}$, meaning that the spherical particle has non-dimensional radius of unity. Table \ref{table:geometry} lists the geometric parameters relevant to these particles and definition of the drag coefficients.

{
As was discussed in section \ref{sec:convergence},  we have not proven any convergence rates for the direct method introduced. Rather, we demonstrate the convergence rate of the direct method by comparison with the viscoplastic flow simulations in some examples. Fig.~\ref{fig:schematic_grid}(a) illustrates how the comparison is performed: we fix the number of nodes in the Cartesian mesh (i.e.~direct method) and also the FE mesh used in the VP simulations. We first calculate a reference critical plastic drag coefficient ($C_{d,c}^{P,*}$), using 10 cycles of adaptive meshing and using increasing values of $B$, i.e.~as in Fig.~\ref{fig:Mesh}. Taking the reference as our exact solution we now calculate the error between $C_{d,c}^{P}$ and $C_{d,c}^{P,*}$, for varying $n_r$. Fig.~\ref{fig:schematic_grid}(b) shows that, on the same number of nodes, the direct method is in fact superior from the perspective of accuracy.}

{ In Fig.~\ref{fig:schematic_grid}(c) we compare the error of the direct method between an ellipsoid and a cylinder. Convergence is poorer for an ellipsoid with the same aspect ratio which suggests that Cartesian mesh is not the best choice for simulating such non-flat particle shapes. As we can see the experimental convergence rate depends on the particle shape: $O(h^{0.8})$ for an ellipsoid and $O(h^{2.5})$ for a cylinder, both with aspect ratio $\rchi=2$ and where $h$ is the mesh size. Dependence of the convergence rate on the orientation of interfaces with respect to that of the mesh is a known phenomenon in the finite difference discretization of the total variation of scalar functions, see \cite{CaiCha20} for a detailed analysis in that case. Equally, for the viscoplastic computations we have only to observe the error in Fig.~\ref{fig:schematic_grid}(b) to see the contribution of mesh adaptivity to calculating the yield limit.}

\begin{figure}
\centering
\includegraphics[width=0.8\textwidth]{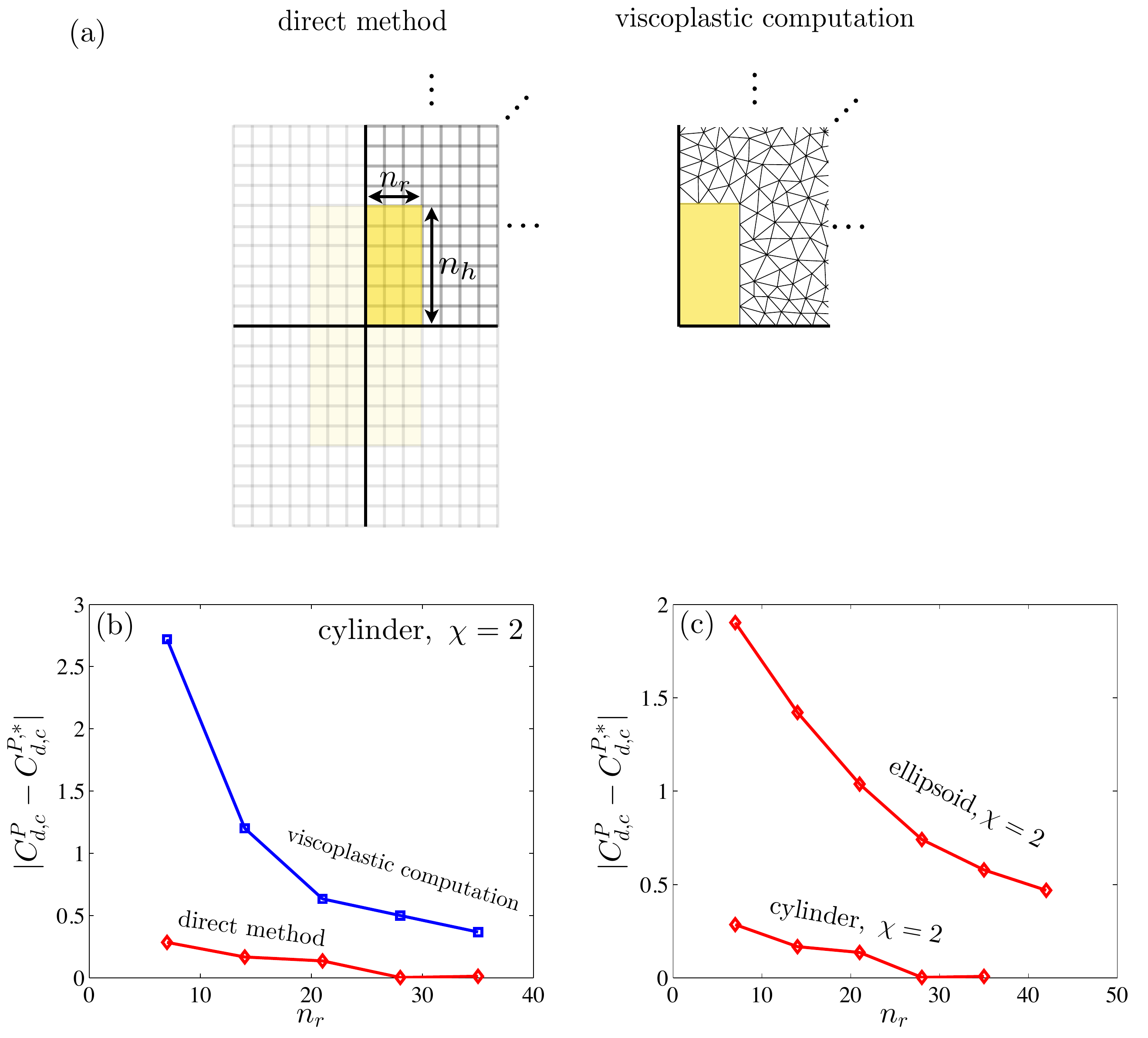}
\caption{{(a) Schematic of the grids in the direct method and the viscoplastic computations,~(b) The panel shows the convergence with respect to the number of the nodes in the radial direction of the particle ($n_r$ as shown in the schematic). Please note that $C_{d,c}^{P,*}$ is the critical plastic drag coefficient extracted from the viscoplastic simulations with the adaptive mesh as the reference point (i.e.~26.1),~(c) Direct method predictions for an ellipsoid and a cylinder with aspect ratio $\rchi=2$. Note that since the aspect ratios are equal to 2, hence $n_h=2 ~n_r$ in these simulations.}}
\label{fig:schematic_grid}
\end{figure}

\begin{table}
  \begin{center}
\begin{tabular}{ l | c | c | c | c }
\hline
\hline
  & Cylinder (axial) & Ellipsoid & Cylinder (trans.) & Parallelepiped \\
\hline
\hline
$\hat{A}_{\bot}$ & $\pi \hat{a}_c^2$ & $\pi \hat{a}_e^2$ & $4 \hat{a}_c \hat{b}_c$ & $4 \hat{a}_p^2$ \\
$\hat{\ell}_{||}$ & $2\hat{b}_c$ & $2\hat{b}_e$ & $2 \hat{a}_c$ & $2 \hat{b}_p$ \\
$\hat{V}$ & $2 \pi \hat{a}_c^2 \hat{b}_c$ & $\frac{4}{3} \pi \hat{a}_e^2 \hat{b}_e$   & $2 \pi \hat{a}_c^2 \hat{b}_c$ & $8 \hat{a}_p^2 \hat{b}_p$ \\
$\hat{L}$ & $\left( \frac{3}{2} \hat{a}_c^2 \hat{b}_c  \right)^{\frac{1}{3}}$ & $\left( \hat{a}_e^2 \hat{b}_e  \right)^{\frac{1}{3}}$ & $\left( \frac{3}{2} \hat{a}_c^2 \hat{b}_c  \right)^{\frac{1}{3}}$ & $\left( \frac{3}{2\pi} \hat{a}_p^2 \hat{b}_p  \right)^{\frac{1}{3}}$\\
$\rchi$ & $\frac{\hat{b}_c}{\hat{a}_c}$ & $\frac{\hat{b}_e}{\hat{a}_e}$ & $\frac{\hat{b}_c}{\hat{a}_c}$ & $\frac{\hat{b}_p}{\hat{a}_p}$ \\
$A_{\bot}$ & ~ $\pi \left( \frac{3 \rchi}{2} \right)^{-\frac{2}{3}}$& $\pi \rchi^{-\frac{2}{3}}$ & $4 \left( \frac{2}{3} \right)^{\frac{2}{3}} \rchi^{\frac{1}{3}}$ & $4 \left( \frac{2\pi}{3} \right)^{\frac{2}{3}} \rchi^{-\frac{2}{3}}$ \\
$\ell_{||}$ & $2 \left( \frac{2}{3} \right)^{\frac{1}{3}} \rchi^{\frac{2}{3}}$ & $2 \rchi^{\frac{2}{3}}$ & $2 \left( \frac{2}{3} \right)^{\frac{1}{3}} \rchi^{-\frac{1}{3}}$ & $2 \left( \frac{2\pi}{3} \right)^{\frac{1}{3}} \rchi^{\frac{2}{3}}$ \\
\hline
\end{tabular}
\caption {Dimensional and dimensionless parameters for the geometries considered. We have kept the definition of the aspect ratio $\rchi$ for the cylinder falling transversally to be that of the same (rotated) object, which modifies the scaling of $A_\bot$ and $\ell_{||}$ in terms of $\rchi$.}\label{table:geometry}
\end{center}
\end{table}

Table \ref{table:CdcComparison} presents a comparison between the critical drag coefficients from the axisymmetric computations {(with adaptive meshing)} and the direct method, calculated over a wide range of aspect ratios $\chi$. Both computations agree reasonably well for intermediate $\chi$ computed, but there are significant errors at both small and large aspect ratios.

The discrepancies are evidently numerical { in origin. The size $n = 200$ denotes the number of points in the longest axial direction of the computational domain. We have also computed the direct method for $n=120,~160$, which establishes that the computation of $C^P_{d,c}$ has converged within mesh resolution.} However, the meshing used for the 3D code is completely uniform and rectangular; we also determine the limiting ranges of the mesh heuristically. Although the mesh domains are large enough to contain the yielded envelopes of the particle, as we know from 2D and axisymmetric computations, yield surfaces have distinct geometric features. These become important in representing the yield limit.  { As our computations are limited in size (i.e.~$n$), as the yielded domain stretches for either small or large $\chi$, the mesh spacing increases and it may be that the regular mesh cannot represent these geometric features} effectively in either asymptotic limit. We continue to explore new discretizations for the 3D code.

\begin{table}
\begin{center}
\begin{tabular}{ l | c | c | c | c}
  \hline
  \hline
  ~ & Cylinder (axial) & Cylinder (axial) & Ellipsoid    & Ellipsoid    \\
    & Viscoplastic     & Direct $n=200$   & Viscoplastic & Direct $n=200$\\
  \hline
  \hline
$\rchi=0.02$ & 12.36  &  9.32 & 12.15 & 8.74    \\
$\rchi=0.14$ & 13.43  & 13.15 & 12.17 & 11.66   \\
$\rchi=0.5$  & 16.39  & 16.15 & 12.78 & 12.92   \\
$\rchi=1$    & 20.01  & 19.73 & 14.19 & 14.51   \\
$\rchi=2$    & 26.10  & 26.09 & 17.61 & 18.25   \\
$\rchi=5$    & 42.02  & 42.22 & 29.06 & 30.03   \\
$\rchi=10$   & 65.76  & 65.77 & 46.02 & 49.03   \\
$\rchi=50$   & 238.50 & 237.50 & 185.00 & 195.24\\
  \hline
\end{tabular}
\caption{Critical drag coefficient $C^P_{d,c}$ for different particles, computed with the methods of \S \ref{sec:axisymmetric_calcs} (Viscoplastic) and of \S \ref{sec:minimization} (Direct).}
\label{table:CdcComparison}
\end{center}
\end{table}

 { \subsection{3D particles without axial symmetry}}

We have also used the direct method to compute some non-axisymmetric particle shapes. Computed estimates of $Y_c$ are given in Table \ref{table:YcDirect}. We use $Y_c$ as opposed to the drag coefficients as the cross-sectional areas are not comparable e.g. for axially and transversely oriented cylinders. We also have computed flows around parallelepipeds at various aspect ratios $\rchi$. Given the above comments, we expect those results at $\rchi \sim O(1)$ to be the most accurate and useful as future benchmarks.


\begin{table}
\begin{center}
\begin{tabular}{ l | c | c | c | c }
  \hline
  \hline
  ~ & Cylinder (trans.) & Cylinder (axial) & Parallelepiped & Ellipsoid \\
  \hline
  \hline
$\rchi=0.14$ & 0.0895 & 0.0359 & 0.0332 & 0.0308 \\
$\rchi=0.5$  & 0.0911 & 0.0682 & 0.0636 & 0.0650 \\
$\rchi=1$    & 0.0851 & 0.0885 & 0.0827 & 0.0919 \\
$\rchi=2$    & 0.0759 & 0.1063 & 0.0998 & 0.1160 \\
$\rchi=5$    & 0.0606 & 0.1210 & 0.1139 & 0.1298 \\
$\rchi=10$   & 0.0494 & 0.1233 & 0.1161 & 0.1262 \\
  \hline
\end{tabular}
\caption{$Y_c$ for different particles, computed with the method of Section \ref{sec:minimization} (Direct) and $n=200$.}
\label{table:YcDirect}
\end{center}
\end{table}

Examples of some of these computed limiting flows are shown below in Fig.~\ref{fig:3dsolutions}. We show an approximation to the limiting plug shape around the particle (where {$|v+(0,0,1)| \ls 0.1$}). As can be seen the method is able to compute non-convex particle shapes as well.

\begin{figure}
  \centering
  \includegraphics[height = 0.225 \paperheight]{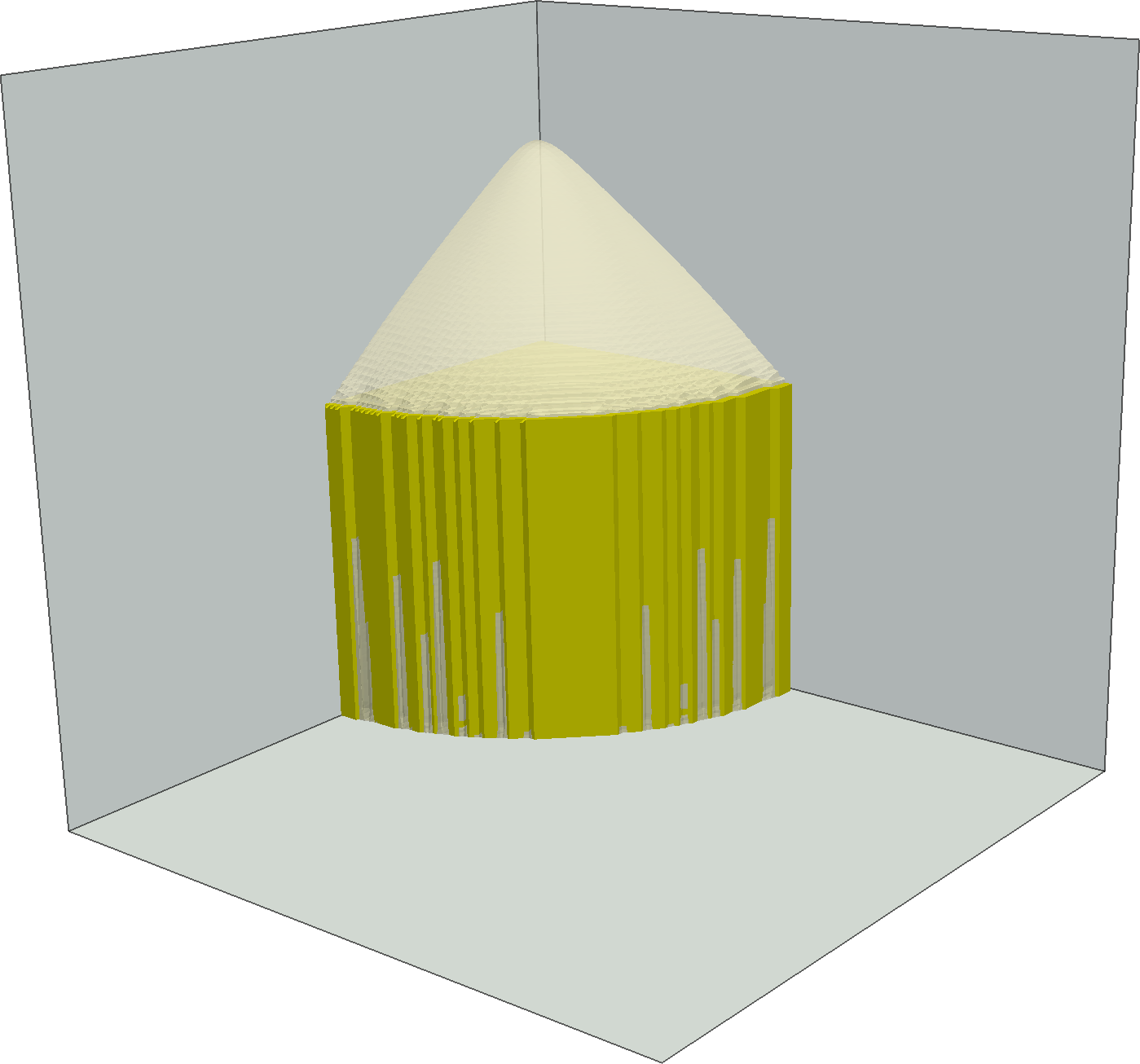}
    \hspace{.2cm}
  \includegraphics[height = 0.225 \paperheight]{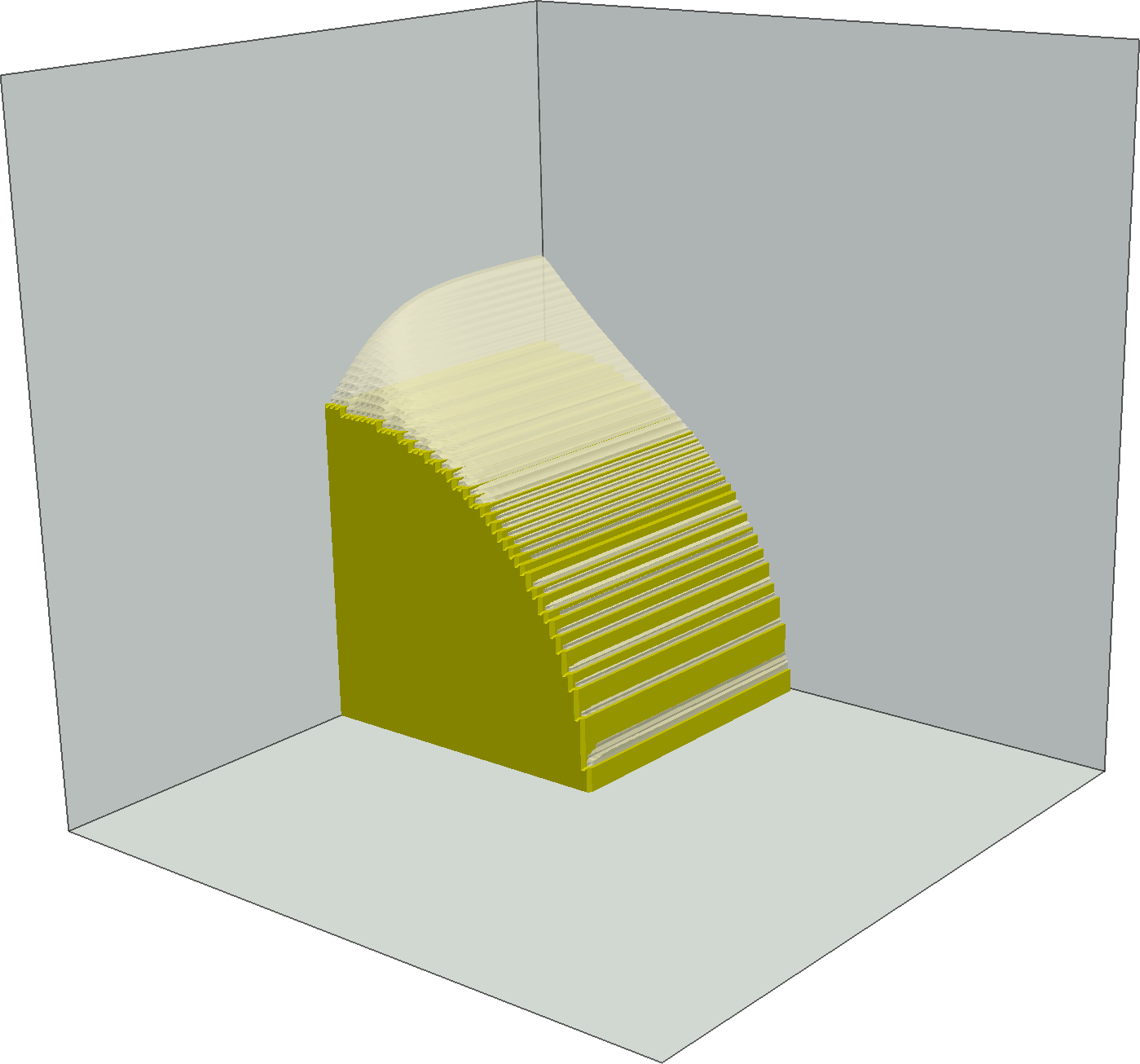}
  \\
  \includegraphics[height = 0.225 \paperheight]{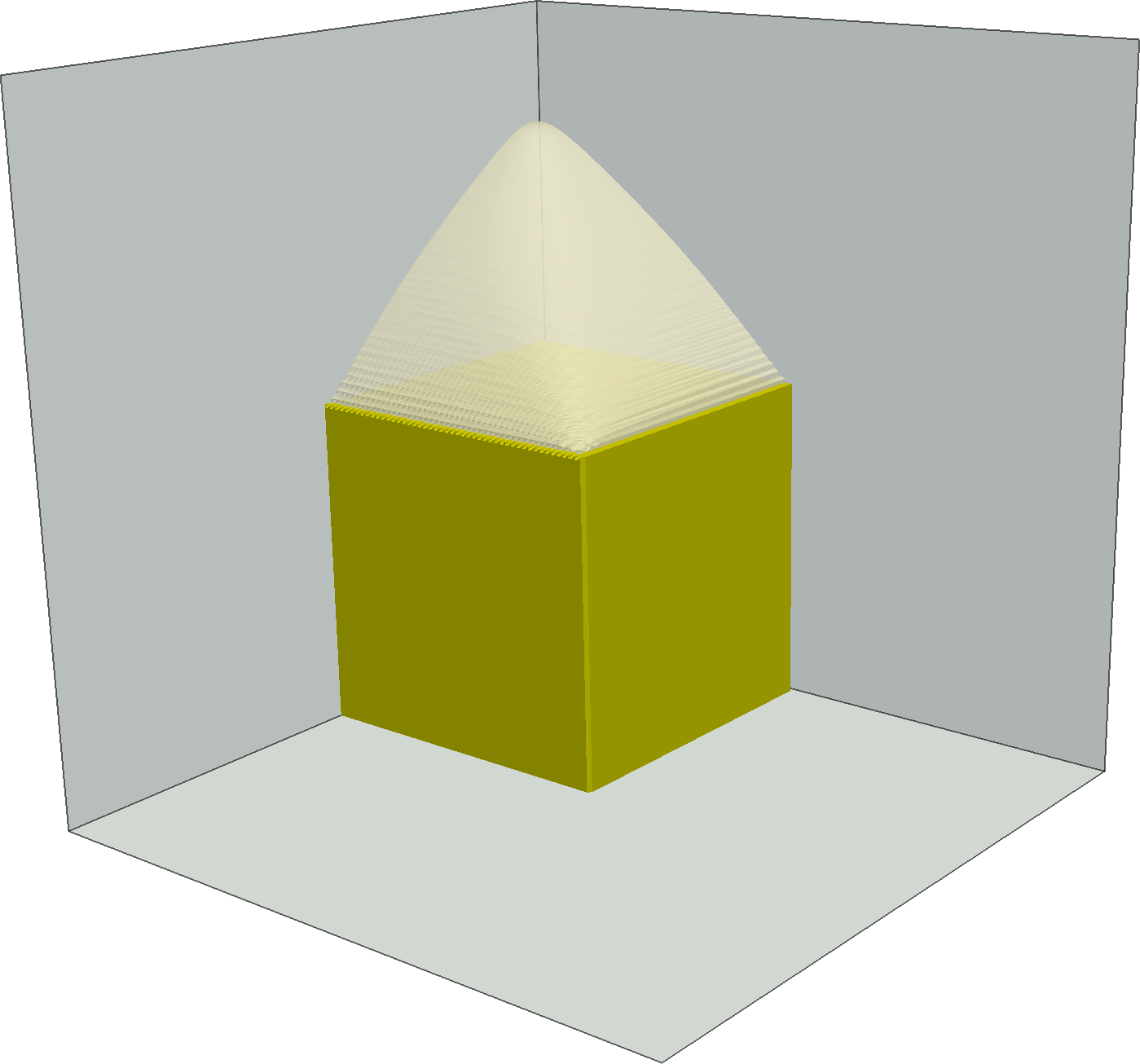}
    \hspace{.2cm}
  \includegraphics[height = 0.225 \paperheight]{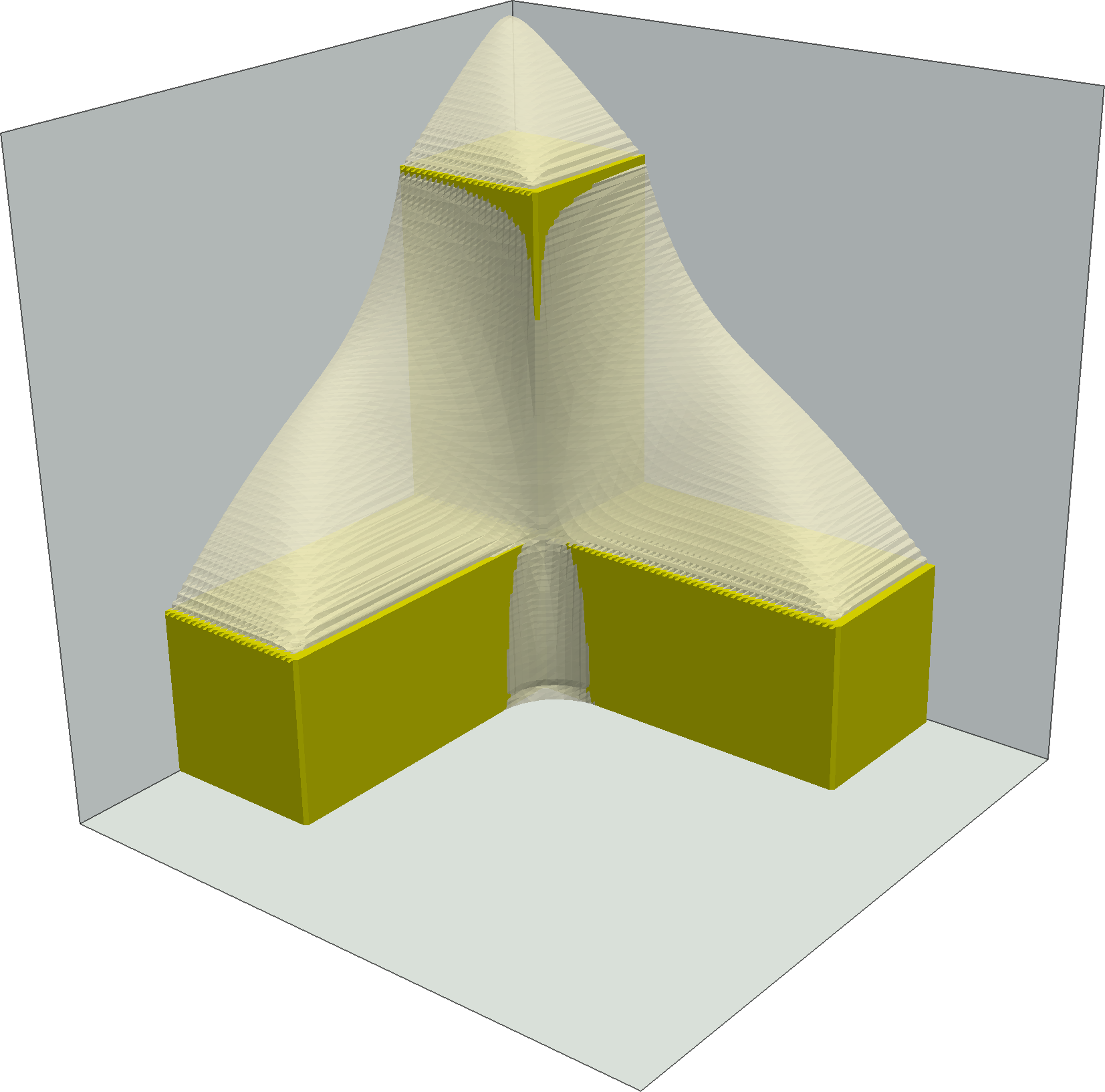}
  \caption{Some computed 3D limiting plug shapes, defined as the surface {$\partial \{ |v+(0,0,1)| \ls 0.1 \}$}. The cylinders and cube correspond to $\rchi=1$, $n=200$. Computations are done in an octant to exploit the symmetry. For the last object, the computed $Y_c$ is $0.0677$, with $n=180$. The values for the others (also with different aspect ratios) can be found in Table \ref{table:YcDirect}.}
  \label{fig:3dsolutions}
\end{figure}

{In 2D flows, different particle shapes can be hidden in the same unyielded envelope; a phenomenon called cloaking \cite{chaparian2017cloaking}. Here we examine possibility of such a phenomenon in 3D by studying a dumbbell-shaped particle. Fig.~\ref{fig:dumbbell} shows the schematic of the two dumbbells considered: two cubic ends with a bridging element of cylindrical (panel a) and parallelepiped shape (panel b). As it is clear from the bottom panels of Fig.~\ref{fig:dumbbell}, both particles share the same unyielded envelope since the bridging part is cloaked, which is intuitive. Hence, the critical plastic drag coefficients are the same: with the circular bridge and $n=200$ we obtain $C_{d,c}^P = 25.5043$ and with parallelepiped bridge $C_{d,c}^P =25.5026$. In comparison, for the complete paralellepiped of aspect ratio $\rchi=2$ we obtain the marginally higher value $C_{d,c}^P =25.6109$, consistent with the slight concavity of the plugs for the dumbbells.}

\begin{figure}
\centering
\includegraphics[width=0.6\textwidth]{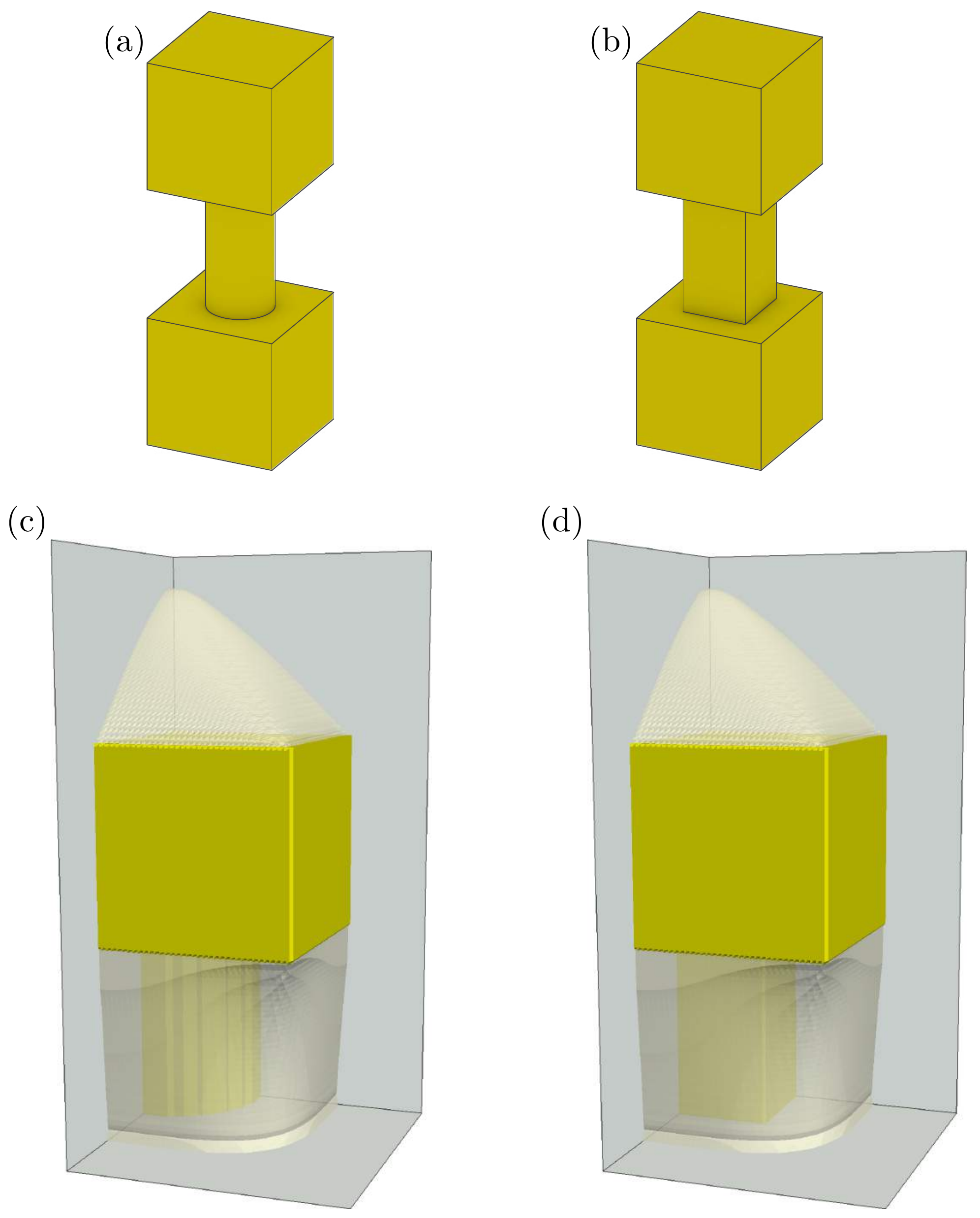}
\caption{{Limiting flows about two dumbbells with cubic ends: (a,b) schematic of the particle shape: one with a cylindrical and the other one with a parallelepiped bridge,~(c,d) computed 3D limiting plug shapes.}}
\label{fig:dumbbell}
\end{figure}

{ Finally, we comment about the unyielded caps sitting on the leading and trailing edges of the particles. As can be predicted by the theory of perfectly-plastic materials \cite{hill1950}, in 2D the unyielded envelope makes $\pi/4$ angle with the vertical symmetry line of the particle shape. This can be concluded from the fact that the shear stress on the vertical symmetry line is zero and the characteristic lines of the hyperbolic governing equations in 2D, which are the maximum shear stress directions as well, make a $\pi/4$ angle with this line. In 3D, the axis of symmetry is again a principal stress direction (i.e.~$\tau_{rz}=0$ at $r=0$), but the extension of the method of characteristics to 3D is not clear since the governing equations are not hyperbolic, due to the presence of hoop stress \cite{hill1950}. However, it is worth mentioning that since the velocity in the radial direction is zero in the unyielded caps attached to the leading and trailing edges of the particles, the hoop stress is zero there and the equations are locally hyperbolic. This could be the main reason that in all the simulations presented in Figs.~\ref{fig:3dsolutions} and \ref{fig:dumbbell}, the unyielded caps make a $\pi/4$ angle with the axial symmetry line.}

\section{Conclusions}
\label{sec:discussion}

Although the yield limit in many classical yield-stress fluid mechanics problem is well-documented as the maximization/minimization problems \cite{frigaard2019background}, practical methods for capturing this important limit are mostly based on the full computations of the Stokes equations. This involves computing the flow in a yield stress fluid and iteratively approaching the yield limit by increasing Bingham/yield numbers in [R]/[M] formulations. To ensure an accurate computation, augmented Lagrangian methods and mesh adaptation should be utilized in these methods, which are relatively slow. Especially when it comes to 3D problems, the cost dramatically increases. Other less popular methods such as the method of characteristics, {\it slipline solutions}, in perfectly-plastic mechanics also may be used to study the yield limit \cite{chaparian2017cloaking,dubash2009final,chaparian2018box}. Nevertheless, this method has some intrinsic drawbacks: finding sliplines is not trivial for complex problems and for some cases the admissible stress and velocity fields associated with the sliplines yield inexact solutions, i.e.~a large uncertainty between lower and upper bounds of the ``load limits''. In the present study, we have presented an alternative method for calculating $Y_c$ by direct usage of the maximization/minimization problem which defines $Y_c$. It can be summarized in two main steps as follows:

\begin{itemize}
\item[(i)] We established that there is a solution in a subspace of BD, the space of possibly discontinuous vector fields of bounded deformation (Thm \ref{thm:relaxation}) and looked at the ways to approach this directly using suitably scaled velocities: solutions for \eqref{eq:quotient} can then be recovered as limits of rescaled physical velocities (Thm \ref{thm:profile}). We explore convex duality for this problem, deriving two dual formulations: one without explicit hydrostatic contributions, easier to analyze in the continuous setting, and another with those contributions, which serves as a starting point for the discretized problem we use.
\item[(ii)] We discretize the saddle point formulation corresponding to the dual problem with hydrostatic contributions through the use of an upwind finite difference scheme particularly adapted to min-max problems with discontinuous solutions. To minimize this discrete formulation we apply a primal-dual hybrid gradient method with proximal regularization in both steps, which makes it suitable for problems with no strict convexity or smoothness.
\end{itemize}

Moreover, we extended our previous studies addressing the yield limit of 2D particles \cite{chaparian2017yield,chaparian2017cloaking} to axisymmetric particles, namely cylinders and ellipsoids. This is feasible via conducting numerical simulations of VP fluid flows about particles in the meridian plane taking into account the hoop stress by introducing an extra component to the Lagrange multiplier ({\it true} stress field upon suitable convergence of the augmented Lagrangian method).

We have validated the 3D direct method for calculating $Y_c$ by comparing with the VP computations. If the particle size in the gravity-plane is comparable with the height of the particle (i.e. intermediate aspect ratios $\rchi$ close to unity), the results are closely matched. In the limit of large/small aspect ratios, however, there is a small discrepancy between the two methods, This may arise because of different meshing strategies in these two distinct methods: anisotropic mesh adaptation is utilized in the axisymmetric VP computations, whereas in the direct method, a regular Cartesian mesh is used for discretization. Indeed, in these limits, the particle shapes become harder to resolve with uniform grids in 3D.
For the present method we have proven (Thm \ref{thm:conv} and Cor \ref{cor:convmin}) that the minimizers of the discrete problems converge to minimizers of \eqref{eq:quotient} as the discretization is refined, but without convergence rates, which are in general not available for such problems without smooth terms. Although an apparent  disadvantage with respect to the methods using the viscoplastic rheology, it is worth noting that convergence rates for most of these methods will also in general degenerate as the yield limit is approached.

With regard to perspectives, we postpone proposing a more suitable/accurate discretization of the primal dual problem \eqref{eq:contprimaldual} to a future study, e.g.~using adaptive finite element schemes. Perhaps more pressing is to extend the range of flows considered. Here we have performed some sample computations with the direct method to calculate $Y_c$ for complex 3D particles, which cannot be computed axisymmetrically. These illustrate one of the benefits of the direct method: no restrictions on the particle shape \emph{per se}. However, on closer inspection we have indeed restricted the limiting problem to linear particle motions. In more generality we would expect also a rotational motion of a non-symmetric particle and these flows need careful examination.

In terms of applicability of the proposed method beyond calculating  particle yield limits, we comment that
the discretization and optimization methods used here can also be applied to problems with different boundary conditions, such as those appearing in landslide predictions \cite{HilIonLacRos02}. In fact, since we can work directly with the yield limit regardless of the viscous rheology appearing in the yielded regions, the same method can be used for limit load analysis in plasticity.

\section*{Acknowledgements}
This work has been supported by the Austrian Science Fund (FWF) within the national research network `Geometry+Simulation', project S11704.
Part of the research has been carried out at the University of British Columbia, supported by Natural Sciences and Engineering Research Council of Canada via their Discovery Grants programme (Grant No. RGPIN-2015-06398).

\bibliography{RefParticle}

\appendix
\section{Proof of Theorem \ref{thm:conv}}\label{sec:conv}
We modify the proofs for the anti-plane case in \cite{IglMerSch18b} to account for the vectorial nature of the problem and the divergence constraint. For the latter, we will need the following discrete lemma:
\begin{lem}\label{lem:discdivproj}
There exists a linear operator $P: S^3 \to S^3$ such that for all $\mathbf v$ we have \[\div_C^n P \mathbf v = 0\text{,  }(P \mathbf v)^{ijk} = (0, 0, 0)\text{ if }(i,j,k) \in G^n \setminus \Omega^n \text{,  }(P \mathbf v)^{ijk} = (0, 0, 1)\text{ if }(i,j,k) \in X^n,\]
and
\begin{equation}\label{eq:projest}\|\mathbf v - P \mathbf v\|_2 \ls C \| \div_C^n P \mathbf v \|_2,\end{equation}
where the constant $C$ depends on $\Omega$ and $X$ but not on $n$.
\end{lem}
\begin{proof}
It follows essentially by decomposing $S^3$ into two orthogonal complementary subspaces. The notion of orthogonality used arises from the discrete integration by parts formula for arbitrary $\mathbf v \in S^3$ and $u \in S$ with zero values on $\left( G^n \setminus \Omega^n \right) \cap X^n$:
\[\sum_{(i,j,k)\in \Omega^n \setminus X^n} \nabla^n_C u \cdot \mathbf v = - \sum_{(i,j,k)\in \Omega^n \setminus X^n} u \div_C^n \mathbf v.\]
For detailed proofs, see \cite[Thm.~2]{Cho69} for the case of centered differences and periodic boundary conditions and \cite[Thms.~2.2, 2.3]{KurSog18} for Dirichlet boundary conditions and forward differences.
\end{proof}
We will also need the following continuous lemma:
\begin{lem}\label{lem:tdpos}
Let $\mathbf{v} \in \BD(\R^d)$ and $\Omega \subset \R^d$ open. Then $\TD(\mathbf{v},\Omega) = \TD^+(\mathbf{v},\Omega)$ defined by
\begin{equation}\TD^+(\mathbf{v},\Omega):=\sup \left\{\int \mathbf{v} \cdot \div \mathbf{q}_1 - \mathbf{v} \cdot \div \mathbf{q}_2 \ \middle \vert \  \mathbf{q}_1,\mathbf{q}_2 \in \mathcal C_0^1(\Omega, \R^{d \times d}_{\mathrm{sym}}),\, |\mathbf{q}_1|_F^ 2+|\mathbf{q}_2|_F^2\ls 1,\, \mathbf{q}_1,\mathbf{q}_2 \gs 0 \right \}. \label{eq:TDpos} \end{equation}
where $\mathbf{q}_1,\mathbf{q}_2 \gs 0$ indicates that each component of $\mathbf{q}_1$ and $\mathbf{q}_2$ is nonnegative.
\end{lem}
\begin{proof}
We recall that
$$\TD(\mathbf{v},\Omega) = \sup\left \{ \int \mathbf{v} \cdot \div \mathbf{q} \ \middle \vert \ \mathbf{q} \in \mathcal C_0^1(\Omega,\R^{d \times d}_{\mathrm{sym}}),\, |\mathbf{q}| \ls 1 \right \}.$$
Let $\mathbf{q}_1,\mathbf{q}_2$ be admissible in the right hand side of \eqref{eq:TDpos}. Then we notice that $\mathbf{q}_1-\mathbf{q}_2$ is admissible in the above, because $\mathbf{q}_1,\mathbf{q}_2$ being componentwise {nonnegative} implies
$$|\mathbf{q}_1-\mathbf{q}_2|_F^2=|\mathbf{q}_1|_F^2+|\mathbf{q}_2|_F^2-2\,( \mathbf{q}_1:\mathbf{q}_2)\ls 1-2\, (\mathbf{q}_1:\mathbf{q}_2) \ls 1,$$
and since $\div (\mathbf{q}_1-\mathbf{q}_2)=\div \mathbf{q}_1 - \div \mathbf{q}_2$ we have
$$\TD^+(\mathbf{v},\Omega) \ls \TD(\mathbf{v},\Omega).$$

To prove the reverse inequality, let $\eps > 0$ be arbitrary and $\mathbf{q}_\eps \in \mathcal C_0^1(\Omega,\R^{d \times d})$ such that
$$\TD(\mathbf{v},\Omega) - \int_{\Omega} \sum_{j=1}^d v_j \div (\mathbf{q}_\eps)_j < \eps,$$
which we can write (renaming $\mathbf{q}_\eps$ to its additive inverse, for convenience) as
\begin{equation}\label{eq:prodineq}\int_{\Omega} \left(1 - \mathbf{q}_\eps : \frac{\d \dot{\bs \gamma}(\mathbf{v})}{\d| \dot{\bs \gamma}(\mathbf{v})|_{F}} \right) \d| \dot{\bs \gamma}(\mathbf{v})|_{F} < \eps,\end{equation}
where $\frac{\d  \dot{\bs \gamma}(\mathbf{v})}{\d| \dot{\bs \gamma}(\mathbf{v})|_F}$ is the Radon-Nikodym derivative of the tensor-valued measure $\dot{\bs \gamma}(\mathbf{v})$ with respect to the unsigned scalar measure $| \dot{\bs \gamma}(\mathbf{v})|_F$. Noting that $|\mathbf{q}_\eps|_F \ls 1$ and $\frac{\d  \dot{\bs \gamma}(\mathbf{v})}{\d| \dot{\bs \gamma}(\mathbf{v})|_F} \ls 1$, the last inequality implies (since for $|\bs{\mu}|_F,|\bs{\nu}|_F \ls 1$, $|\bs{\mu} - \bs{\nu}|_F^2 \ls 2 - 2 \bs{\mu} : \bs{\nu}$) {that}
\begin{equation}\label{eq:distineq}\int_{\Omega} \frac 12 \left \vert  \mathbf{q}_\eps- \frac{\d \dot{ \bs \gamma}(\mathbf{v})}{\d|\dot{ \bs \gamma}(\mathbf{v})|_F}  \right \vert^2 \dd|\dot{ \bs \gamma}(\mathbf{v})|_F < \eps.\end{equation}
Notice also that by \eqref{eq:prodineq} and the Cauchy-Schwarz inequality we have
\begin{equation}\label{eq:closetoone}
\begin{aligned}\int_{\Omega}1-|\mathbf{q}_\eps|_F^2 \dd|\dot{ \bs \gamma}(\mathbf{v})|_F&=\int_{\Omega}\big(1-|\mathbf{q}_\eps|_F\big)\big(1+|\mathbf{q}_\eps|_F\big)\dd|\dot{ \bs \gamma}(\mathbf{v})|_F \leq 2\int_{\Omega}1-|\mathbf{q}_\eps|_F \dd|\dot{ \bs \gamma}(\mathbf{v})|_F\\
&\leq 2 \int_{\Omega} \left( 1 - \mathbf{q}_\eps : \frac{\d \dot{ \bs \gamma}(\mathbf{v})}{\d|\dot{ \bs \gamma}(\mathbf{v})|_F} \right) \d|\dot{ \bs \gamma}(\mathbf{v})|_F <2\eps,
\end{aligned}
\end{equation}

Now we replace the components $(q_\eps)_{ij}$ by $(\tilde q_\eps)_{ij}$ which are smooth, coincide with $(q_\eps)_{ij}$ out of $\{|(q_\eps)_{ij}|_F < \sqrt{\eps} \}$, that satisfy
$$|\tilde{\mathbf{q}}_\eps|_F \ls |\mathbf{q}_\eps|_F,$$
and such that $\{(\tilde q_\eps)_{ij} = 0\}$ is the closure of an open set.
One can for example choose
$$0 < \alpha < \min\left(\frac{1}{\sqrt{d}}, \sqrt{\eps}\right)$$
and define a smooth nondecreasing function $\psi_\alpha : \R \to \R$ such that $\psi_\alpha(t) = t$ for $|t| \gs \alpha$, $|\psi_\alpha(t)| \leq |t|$ and $\psi_\alpha (-\alpha/2, \alpha/2)=\{0\}$
to define
$$(\tilde q_\eps)_{ij} := \psi_\alpha \circ (q_\eps)_{ij}.$$

Thus we have $|\tilde{\mathbf{q}}_\eps|_F \ls 1$ and $|(\tilde q_\eps)_{ij} - (p_\eps)_{ij} | \ls 2\sqrt{\eps}$, and taking into account \eqref{eq:prodineq} we obtain
\begin{equation}\label{eq:tildedistineq}\begin{aligned}\Biggl( \int_{\Omega} \left \vert  \tilde{\mathbf{q}}_\eps- \frac{\d \dot{ \bs \gamma}(\mathbf{v})}{\d|\dot{ \bs \gamma}(\mathbf{v})|_F}  \right \vert_F^2 \dd|\dot{ \bs \gamma}(\mathbf{v})|_F \Biggr)^{\frac 12} &\ls \left( \int_{\Omega} \left \vert \mathbf{q}_\eps - \frac{\d \dot{ \bs \gamma}(\mathbf{v})}{\d|\dot{ \bs \gamma}(\mathbf{v})|_F}  \right \vert_F^2 \dd|\dot{ \bs \gamma}(\mathbf{v})|_F \right)^{\frac 12} + \left( \int_{\Omega} |\tilde{\mathbf{q}}_\eps - \mathbf{q}_\eps |_F^2 \dd|\dot{ \bs \gamma}(\mathbf{v})|_F \right)^{\frac 12} \\
&\ls C \sqrt{\eps} \left( 1 + |\dot{ \bs \gamma}(\mathbf{v})|_F(\Omega)\,\right).
\end{aligned}\end{equation}

Furthermore, using \eqref{eq:closetoone} and the definition of $\tilde{\mathbf{q}}_\eps$ we obtain the estimate
\begin{equation*}\label{eq:tildeclosetoone}
\int_{\Omega}1-|\tilde{\mathbf{q}}_\eps|_F^2 \dd|\dot{ \bs \gamma}(\mathbf{v})|_F = \int_{\Omega}1-|\mathbf{q}_\eps|_F^2 \dd|\dot{ \bs \gamma}(\mathbf{v})|_F + \int_{\Omega} | \mathbf{q}_\eps |_F^2 - |\tilde{\mathbf{q}}_\eps |_F^2 \dd|\dot{ \bs \gamma}(\mathbf{v})|_F \leq 2\eps + 4d\eps|\dot{ \bs \gamma}(\mathbf{v})|_F < C\eps(1+|\dot{ \bs \gamma}(\mathbf{v})|_F),
\end{equation*}
which ensures, writing $1-\mu:\nu = \frac 12 (1-|\mu|_F^2 + 1- |\nu|_F^2 + |\mu - \nu|_F^2)$ and by \eqref{eq:tildedistineq} that
$$\left \vert \int_{\Omega} \left(\tilde{\mathbf{q}}_\eps : \frac{\d \dot{ \bs \gamma}(\mathbf{v})}{\d|\dot{ \bs \gamma}(\mathbf{v})|_{F}} -1 \right) \dd|\dot{ \bs \gamma}(\mathbf{v})|_F \right \vert \ls C \eps\left( 1 + \big(\ 1 + |\dot{ \bs \gamma}(\mathbf{v})|_F(\Omega)\, \big)^2\right).$$
Now, we notice that having fattened the level-set $\{(\tilde{\mathbf{q}}_\eps)_{ij}=0\}$, we can write
$$ (\tilde{\mathbf{q}}_\eps)_{ij} = \left[(\tilde{\mathbf{q}}_\eps)_{ij}\right]^+ - \left[(\tilde{\mathbf{q}}_\eps)_{ij}\right]^-$$ where both quantities are smooth.
Writing similarly
$$\tilde{\mathbf{q}}_\eps = \tilde{\mathbf{q}}_\eps^+ - \tilde{\mathbf{q}}_\eps^-$$
with $\tilde{\mathbf{q}}_\eps^\pm$ are smooth and have only positive components, we note that $(\tilde{\mathbf{q}}_\eps^+,\tilde{\mathbf{q}}_\eps^-)$ are admissible in the right hand side of \eqref{eq:TDpos}, showing that
$$\TD^+(v,\Omega) \gs \TD(v,\Omega) -C(v)\eps.$$
Making $\eps \to 0$, we conclude.
\end{proof}
We are now ready to prove the convergence itself.
\begin{proof}[Proof of Theorem \ref{thm:conv}]
First, we study the $\Gamma$-liminf and assume that $\mathbf{v}^n \to \mathbf{v}$ in $L^1$. Notice that we can write $\TD^n(\mathbf{v}^n)$ as a dual formulation
$$\TD^n(\mathbf{v}^n) = \sup\left\{ \sum_{(i,j,k) \in G^n} \left(\mathbf{v}^n\right)^{ijk} \cdot \left(\div^n \mathbf{q}\right)^{ijk} \, \mid \, \mathbf{q}:G^n \to (\R^{3 \times 3}_{\mathrm{sym}})^2, \mathbf{q} = 0 \text{ on } X^n \text{ and } G^n \setminus \Omega^n \right\}$$
where $\div^n \mathbf{q}: G^n \to \R^3$ is defined by
\begin{align*}((\div^n \mathbf{q} )_{\alpha})^{ijk} & := (\mathbf{q}_{1,\alpha,+})^{i,j,k} - (\mathbf{q}_{1,\alpha,+})^{i-1,j,k} + (\mathbf{q}_{1,\alpha,-})^{i,j,k} - (\mathbf{q}_{1,\alpha,-})^{i+1,j,k} + (\mathbf{q}_{2,\alpha,+})^{i,j,k} - (\mathbf{q}_{2,\alpha,+})^{i,j-1,k} \\
& + (\mathbf{q}_{2,\alpha,-})^{i,j,k} - (\mathbf{q}_{2,\alpha,-})^{i,j+1,k} + (\mathbf{q}_{3,\alpha,+})^{i,j,k} - (\mathbf{q}_{3,\alpha,+})^{i,j,k-1} + (\mathbf{q}_{3,\alpha,-})^{i,j,k} - (\mathbf{q}_{3,\alpha,-})^{i,j,k+1},\end{align*}
where zero values are used whenever $i-1 \ls 0$ or $i+1 \gs n$, and similarly for $j, k$. This supremum is obtained easily by a (discrete) integration by parts using the expression
$$|\dot{ \bs \gamma}(\mathbf{v}) \vee 0|_F = \sup_{\substack{|\mathbf{q}|_F \ls 1 \\ q_{ijk} \gs 0}} \dot{ \bs \gamma}(\mathbf{v}):\mathbf{q},$$
and the zero boundary conditions for $\mathbf q$.

Now, we note that every such $\mathbf{q} : G^n \to (\R^{3 \times 3})^2$ can be seen as the discretization of some compactly supported smooth function $\overline{\mathbf{q}}:(0,1)^3 \to (\R^{3 \times 3})^2$, for example stating
$$ \mathbf{q}^{ijk} = \fint_{R_{ijk}} \overline{\mathbf{q}}.$$
As a result, one can write
$$\TD^n(v^n) = \sup\left\{ \mathbf{v}^n \cdot \div^n \overline{\mathbf{q}} \, \middle\vert \, \overline{\mathbf{q}}\in \mathcal C_0^1\left([0,1]^3, (\R^{3 \times 3}_\mathrm{sym})^2\right), \ | \overline{\mathbf{q}}|_F \ls 1, \overline{\mathbf{q}} \gs 0 \ \right\}.$$
Since it just consists of forward and backward finite differences, for a smooth function $\overline{\mathbf{q}}$ the quantity $\div^n \overline{\mathbf{q}}$ converges to
\begin{align*}(\div \overline{\mathbf{q}})_\alpha & = \div(\overline{\mathbf{q}}_{1,\alpha,1},\overline{\mathbf{q}}_{2,\alpha,1}, \overline{\mathbf{q}}_{3,\alpha,1})+ \div(-\overline{\mathbf{q}}_{1,\alpha,2},-\overline{\mathbf{q}}_{2,\alpha,2}, -\overline{\mathbf{q}}_{3,\alpha,2}).
\end{align*}
Therefore, using Lemma \ref{lem:tdpos}
$$\liminf \TD^n(\mathbf{v}^n) \gs \TD^+(\mathbf{v}) = \TD(\mathbf{v}).$$

For $\rchi_{\{\div^n_C = 0\}}$, we have a similar dual characterization: $\div^n_C \mathbf v^n = 0$ if and only if
\begin{equation}\label{eq:discweakincomp}\begin{gathered}\sum_{(i,j,k) \in G^n} (\psi)^{ijk} \left(\div^n_C \mathbf v^n\right)^{ijk} = - \sum_{(i,j,k) \in G^n} \left( \nabla^n_C \psi \right)^{ijk} \cdot (\mathbf v^n)^{ijk} = 0 \\ \text{for all }\psi: G^n \to \R \text{ with }\psi = 0 \text{ on }G^n \setminus \Omega^n \text{ and } X^n.\end{gathered}\end{equation}
Once again, such functions $\psi$ can be seen as discretizations of $\overline \psi \in \mathcal{C}^1_0\big((0,1)^3\big)$, so that if for infinitely many $n$ the condition \eqref{eq:discweakincomp} holds (i.e., $\lim \inf \rchi_{\{\div^n_C\}} ( \mathbf v^n) = 0$), we must have $\div \mathbf v = 0$, since $\mathbf v^n \to \mathbf v$ and $\nabla^n_C \psi \to \nabla \overline \psi$.

Finally, for $\rchi_{C^n}$, let us notice that if $\rchi_{C}(\mathbf{v}) = + \infty$, that is either $\mathbf{v} \not \equiv 0$ on $[0,1]^3 \setminus \Omega$ or $\mathbf{v} \not \equiv 1$ on $X$. If the latter holds, then for $\eps$ small enough, $X \cap \left(\{v_3 > 1+2\eps\} \cup \{v_3 < 1- 2 \eps \} \right)$ has positive measure and thanks to the $L^1$ convergence of $v^n$,
$$\Omega^n_s \cap \left( \{v_3^n >1+ \eps \} \cup \{v_3^n < 1- \eps \} \right)$$
must have a positive measure for $n$ big enough. That implies $\rchi_{C^n}(\mathbf{v}^n) = + \infty$ and the {lim inf} inequality is true. If $\rchi_{C}(\mathbf{v}) < \infty$, then $\rchi_{C}(\mathbf{v}) =0$ and the inequality is also true since $C^n \subset C$.

Let now $\mathbf{v} \in \BD ((0,1)^3)$.  We want to construct a recovery sequence $\mathbf{v}^n \to \mathbf{v}$ such that $$\TD(\mathbf{v})  + \rchi_C(\mathbf{v}) \gs \limsup \TD^n(\mathbf{v}^n) + \rchi_{C^n}(\mathbf{v}^n).$$ If $\mathbf{v} \notin C$, any $\mathbf{v}^n \to \mathbf{v}$ gives the inequality. If $\mathbf{v} \in C$, then we first introduce
$$\mathbf{v}_\delta = \eta_\delta  \ast \mathbf{v}$$
where $\eta_\delta$ is a smooth convolution kernel of width $\delta$. Then, $\TD(\mathbf{v}_\delta) \xrightarrow{\delta \to 0} \TD(\mathbf{v})$ (\cite[Theorem 1.3]{AnzGia80}, noticing that $\mathbf{v}$ is constant around $\partial [0,1]^3$) and, thanks to \eqref{eq:discompact}, if $\delta \ls \frac 1n,$ we have $\rchi_{C^n}(\mathbf{v}_\delta) = 0.$
One can define $\mathbf{v}^n$ by
$$(\mathbf{v}_\delta^n)^{ijk} = \fint_{R_{ijk}^n} \mathbf{v}_\delta, $$
that satisfies $\rchi_{C^n}(\mathbf{v}^n) = 0$, and compute, assuming that $(\mathbf{v}_\delta^n)^{i+1,j,k}-(\mathbf{v}_\delta^n)^{i,j,k} \gs 0$
$$\frac{(\mathbf{v}_\delta^n)^{i+1,j,k}-(\mathbf{v}_\delta^n)^{i,j,k}}{1/n} = \left(\frac 1n\right)^{-1} \fint_{R_{ijk}^n} \mathbf{v}_\delta(x+\frac 1n,y,z) - \mathbf{v}_\delta(x,y,z) \gs \inf_{R_{ijk}^n \cup \big( R_{ijk}^n+(\frac 1n,0,0) \big) } \vert \partial_x \mathbf{v}_\delta \vert . $$
Then since $\mathbf{v}_\delta \in \mathcal C^1$, it is clear that the right hand side converges to $|\partial_x \mathbf{v}_\delta|$.
Note that in the 'upwind' gradient of a smooth function, only one term by direction can be active, then it is also true for $\mathbf{v}^n_\delta$ if $n$ is large enough and therefore $\TD^n(\mathbf{v}^n_\delta) \xrightarrow{n \to \infty} \TD(\mathbf{v}_\delta).$

It remains to take into account the divergence constraint. First we notice that since the convolution with $\eta_\delta$ commutes with weak derivatives, we have that $\div \mathbf v_\delta = 0$. On the other hand, even if $\div \mathbf v = 0$, the functions $\mathbf v^n_\delta$ do not necessarily satisfy $\div^n_C \mathbf v_\delta^n = 0$ exactly. Therefore, we replace them with their discrete divergence-free projection $P \mathbf v_\delta^n$ as defined in Lemma \ref{lem:discdivproj}. If we can prove that $\TD^n(P\mathbf{v}^n_\delta) \xrightarrow{n \to \infty} \TD(\mathbf{v}_\delta)$ as well, we can then conclude by a diagonal argument in $n$ and $\delta$. Now, since $\div \mathbf{v}_\delta = 0$ we have for the centered finite differences that $\|\div^n_C \mathbf v_\delta^n\|_\infty \ls C n^{-2}$ so that the estimate \eqref{eq:projest} becomes
\[\|P \mathbf v_\delta^n - \mathbf v_\delta^n\|_2 \ls C \| \div^n_C \mathbf v_\delta^n \|_2 \ls C n^{3/2} \| \div^n_C \mathbf v_\delta^n \|_\infty \ls C n^{-1/2},\]
where we have used that there are $Cn^3$ points in the grid. For the discrete total deformation we have
\begin{equation}\begin{aligned}\label{eq:tdproj}\left|\TD^n(P\mathbf{v}^n_\delta) - \TD^n(\mathbf{v}^n_\delta)\right| &\ls Cn^{-3} \|\gsym(P\mathbf{v}^n_\delta) - \gsym(\mathbf{v}^n_\delta)\|_1 \ls Cn^{-3/2}\|\gsym(P\mathbf{v}^n_\delta) - \gsym(\mathbf{v}^n_\delta)\|_2 \\ &\ls Cn^{-1/2} \|P\mathbf{v}^n_\delta - \mathbf{v}^n_\delta\|_{2} \ls C n^{-1} \xrightarrow{n \to \infty} 0,\end{aligned}\end{equation}
as desired.
\end{proof}

\begin{cor}\label{cor:convmin}
Minimizers of the discrete minimization problem \eqref{eq:discmin} converge in $L^1$ to minimizers of the total deformation $\TD$ over $\mathbf \in \BD(\R^d)$ such that $\div \mathbf v = 0$ and $\mathbf v = v_0$ on $\Omega \setminus X$ and $\mathbf v = 0$ on $\R^d \setminus \Omega$.
\end{cor}
\begin{proof}
To see this, notice that $\TD^n$ is also the continuous total deformation of the corresponding piecewise constant function, so that proceeding as in Theorem \ref{thm:profile}, if we can check that $\TD^n( \mathbf{v}^n )$ is a bounded sequence, where $\mathbf{v}^n$ are discrete minimizers, up to a subsequence they converge to some limit $\mathbf{v}$ in the $L^1$ topology. To obtain this, just notice that $\TD^n( \mathbf{v}^n ) \ls \TD^n( P\mathbf{z}^n )$, where $\mathbf{z}^n \in K^n$ equals ${-}(0,0,1)$ on all nodes of $G^n \cap \Omega^n$ and $P$ is the projection of Lemma \ref{lem:discdivproj}. That $\TD^n( P\mathbf{z}^n )$ is a bounded sequence can be seen proceeding as in \eqref{eq:tdproj}, since $\TD^n(\mathbf z_ n) = 0$. Standard results in $\Gamma$-convergence then imply that the limit is a minimizer of the $\Gamma$-limit, see \cite{Bra02}.
\end{proof}

\end{document}